\documentclass[11pt]{article}
\usepackage[margin=0.75in]{geometry}

\usepackage{comment}
\usepackage{natbib}
\usepackage{amsthm}
\newtheorem{remark}{Remark}[section]

\newtheorem{proposition}{Proposition}[section]

\usepackage[ruled,vlined]{algorithm2e}
\newtheorem{lemma}{Lemma}[section]

\newtheorem{thm}{Theorem}[section]

\usepackage{multirow}
\usepackage{amsmath}
\usepackage{float}
\usepackage{tabularx}
    \usepackage{color}
\usepackage{graphicx}
\usepackage{ulem}
\usepackage{enumerate}
\usepackage{textcomp}
\usepackage{gensymb}
\usepackage{multicol}
\usepackage{amssymb}

\newcommand{\vectornorm}[1]{\left|\left|#1\right|\right|}
\newcommand{\var}{\text{Var}}

\newcommand{\kl}{\mbox{KL}}

\newcommand{\obs}{_{\text{obs}}}

\newcommand{\lp}{\ensuremath{\left(}}
\newcommand{\rp}{\ensuremath{\right)}}
\DeclareMathOperator*{\argmin}{arg\,min}
\newcommand{\pss}{z}

\newcommand{\prior}{\pi}
\newcommand{\priorb}{\pi^b}

\newcommand{\postb}{p_{\priorb}}

\newcommand{\Wm}{W_2(m)}
\newcommand{\Wmflip}{\widetilde{W}_2(m)}
\newcommand{\norm}[1]{\ensuremath{\mathcal \| #1 \|}}

\begin{document}

\title{Quantifying Observed Prior Impact}
\author{David E. Jones\thanks{Department of Statistics, Taxes A \& M University, College Station, TX, U.S.A.}, Robert N. Trangucci \thanks{Department of Statistics, University of Michigan, Ann Arbor, MI, U.S.A.}, Yang Chen\thanks{Department of Statistics and Michigan Institute for Data Science (MIDAS), University of Michigan, Ann Arbor, MI, U.S.A. Address for correspondence: Yang Chen, 445E West Hall, 1085 South University Avenue, Ann Arbor, MI 48109, USA. E-mail: ychenang@umich.edu} }

\maketitle

\begin{abstract}
We distinguish two questions (i) how much information does the prior contain? and (ii) what is the effect of the prior? Several measures have been proposed for quantifying effective prior sample size, for example \citet{clarke1996} and \citet{Morita:2008}. However, these measures typically ignore the likelihood for the inference currently at hand, and therefore address (i) rather than (ii). Since in practice (ii) is of great concern, \citet{reimherr2014} introduced a new class of effective prior sample size measures based on prior-likelihood discordance. We take this idea further towards its natural Bayesian conclusion by proposing measures of effective prior sample size that not only incorporate the general mathematical form of the likelihood but also the specific data at hand. Thus, our measures do not average across datasets from the working model, but condition on the current observed data. 
Consequently, our measures can be highly variable, but we demonstrate that this is because the impact of a prior can be highly variable. 
Our measures are Bayes estimates of meaningful quantities and well communicate 
the extent to which inference is determined by the prior, or framed differently,  the amount of effort saved due to having prior information. 
We illustrate our ideas through a {number of examples including} a Gaussian conjugate model (continuous observations), a Beta-Binomial model (discrete observations), and a linear regression model (two unknown parameters). Future work on further developments of the   methodology and an application to astronomy are discussed at the end. 
\end{abstract}


\section{Motivation}\label{sec:intro}


Prior knowledge and assumptions are central to many statistical problems, and
in practice it is  important to assess their impact  on the final inference. For example,  \citet{chen2019calibration} propose a Bayesian analysis of a multi-telescope astronomical dataset, and highlight that  scientific prior distributions provide key information about each of the specific instruments and play a substantial role in the final inference. For scientists it is important to understand the role of prior distributions in such scenarios, e.g., do the priors associated with one particular instrument have a much greater impact on the inference than those for other instruments? 

One appealing and interpretable way to assess prior impact is to provide a measure of the {\it effective} prior sample size (EPSS), i.e., the approximate number of observations to which the information in the prior is equivalent.   
Gaussian conjugate models offer a canonical example: with observed data $y_i\overset{iid}{\sim} N(\mu,\sigma^2)$, for $i=1,\dots,n$, and  conjugate prior distribution $\mu \sim N(\mu_0,\sigma^2/r)$, the posterior distribution of $\mu$ is $N(w_n \bar{y}_n + (1-w_n)\mu_0,\sigma^2/(n+r))$, where $w_n=n/(n+r)$. Based on  the posterior variance denominator $n+r$, the effect of the prior appears to be  equivalent to that of $r$ samples, so we say that the EPSS  is $r$. However, this formulation faces  two challenges: (a) it is not immediately clear how to generalize  beyond conjugate models, and more importantly, (b) when $\mu_0$ is arbitrarily different to $\bar{y}_n$, its impact on the posterior mean is arbitrarily large, and is therefore clearly not equivalent to that of $r$ samples. 

Effective prior sample size (EPSS) measures have gained substantial attention in the literature, and a number of strategies have been proposed in response to the two challenges above, e.g.,~\citet{clarke1996}, \citet{Morita:2008}, and~\citet{morita2010evaluating}. Most of the  strategies proposed rely on a comparison between the actual prior $\pi$ and a low-information or baseline prior $\pi^b$, e.g., the improper prior $\pi^b(\mu) \propto 1$ would be a natural choice for the baseline prior in the Gaussian conjugate model above. This comparative information approach is necessary, because there is no universal ``non-informative" prior against which to measure prior impact, and Bayesian inference cannot be conducted without a prior. Early generalizations along these lines sought to match the prior $\pi$ to a hypothetical posterior distribution constructed using the baseline prior $\pi^b$ and some hypothetical previous samples, that is, they  interpreted the  prior $\pi$ as the posterior from a previous analysis. The EPSS is then defined as the number of observations used in the hypothetical posterior distribution, e.g., \citet{clarke1996} and \citet{Morita:2008}. 
These approaches successfully generalize the notion of EPSS, but do not address concern (b) regarding the real impact of the prior when the data mean and prior mean differ substantially. Indeed,  these methods  do not consider the observed data or the real posterior distribution  at all. 

\citet{reimherr2014} instead suggested minimizing the discrepancy between {\it two} posterior distributions, one using the real prior $\pi$ and the other using the baseline prior $\pi^b$. In this case the EPSS is defined as the difference in the number of samples used by the two posteriors.  Similar ideas have also been proposed in slightly different contexts, e.g., see \citet{lin2007information} and \citet{wiesenfarth2019quantification}.  The \citet{reimherr2014} method offers many improvements over early approaches and goes beyond simply capturing the variance of the prior; it also partially quantifies the impact of the prior {\it location}. However, it averages over the data {using the bootstrap}, and therefore does not quantify the impact of the prior for the specific analysis carried out with the observed data at hand, which is of most interest in practice.   

In this paper, we follow a similar approach but propose a new EPSS measure that addresses the above limitations by conditioning on the observed data, and thereby directly quantifies the prior impact for the actual analysis performed. This modification was inspired by the  insight of \citet{efron1978assessing} that observed Fisher information is sometimes more useful than  expected Fisher information. We furthermore provide an explicit definition of EPSS in terms of future observations, and thus identify the estimand of interest. This gives our EPSS measures real-world interpretations and disentangles the tasks of defining and approximating the EPSS. Our method also has additional appealing properties, including a Bayes estimate interpretation and no lower limits on the observed data sample size $n$. The latter is important because prior impact is often most pronounced, and therefore of most interest, when the sample size is small. In contrast, \citet{reimherr2014} require $n$ to be large because their method relies on the bootstrap and an accurate  estimate of the ``true parameter" value, see Section \ref{sec:review} for a review. In summary, 
our approach represents a substantially improved method for quantifying prior impact in practice, and its real-world interpretation could make it a valuable tool for clearly reporting the contribution of priors in Bayesian analyses.

There are important situations and complications that we have not yet addressed, which will be considered in future work. For example,  although our method is general (as will be discussed in Section~\ref{sec:discussion}), in this paper we {focus on} simple conjugate models {to gain intuition}, whereas more general hierarchical models often play an important role in real analyses. Some earlier approaches mentioned above have been extended to more complex hierarchical models, e.g.,  \cite{morita2012prior} extended the ideas of \citet{Morita:2008} to  three-layer models. Our approach as presented here can naturally be applied in such contexts to obtain an overall EPSS, but the most appropriate way to summarize the impact of the prior information related to each individual parameter (as opposed to that related to all the parameters) needs  to be studied further, especially in the multi-level prior context. 

There are alternatives to EPSS  for assessing the impact of a prior, including  sensitivity analysis and direct quantification of  prior information. While also useful, these alternatives have a number of drawbacks. 
Sensitivity analysis  has the advantage of focusing on the actual analysis at hand (like our approach), but is typically somewhat ad hoc, difficult to summarize, and only related to very specific aspects  of the inference. In contrast, prior information measures typically consider the prior in isolation and do not reveal much about the specific analysis performed. It can therefore be challenging for researchers to quantify how much inference is driven by prior information, as opposed to the data at hand. In summary, we focus on EPSS because it is  a very interpretable measure of prior impact, and has the potential to be widely and consistently used, and may thereby provide a much needed assessment of the role of priors in the many studies being performed.   
 
This paper is organized as follows. Section \ref{sec:background} reviews existing approaches of measuring EPSS, discusses the importance of observed data, and defines our EPSS measure. Section \ref{sec:illustration} provides numerical results in the context of a Gaussian conjugate model example. Section \ref{sec:justification_theory} provides intuition and theory supporting our method. Section  \ref{subsec:nonnormal_numerical} provides additional numerical results in the form of  Beta-Binomial and regression model examples. Section \ref{sec:discussion} provides further insights and discussion of open problems. Proofs are given in the Appendix. 


\section{Defining prior information and prior impact}\label{sec:background}

\subsection{Existing methods of measuring effective prior sample size} \label{sec:review} Suppose that $\pi$ is our prior distribution for a  collection of unknown parameters of interest $\theta \in \Theta$. Let $\priorb$ be a baseline prior and  $\boldsymbol{y}=\{y_1,\dots,y_n\}$ be hypothetical previous data with probability density $f(\boldsymbol{y}|\theta)$. Imagine that our real prior $\prior$ is the posterior distribution  $\postb(\cdot|\boldsymbol{y}) \propto f(\boldsymbol{y}|\cdot)\priorb(\cdot)$. Under this formulation,  \citet{clarke1996} considers
\begin{align}
\argmin_{\boldsymbol{y}\in  \mathcal{Y}}  \kl(\prior(\cdot),\postb(\cdot|y)),
\label{eqn:clarke}
\end{align}
where $\kl(g,h)$ denotes the Kullback-Leibler (KL) divergence $\int_\Theta \log(g(\theta)/h(\theta)) g(\theta)d\theta$, and $\mathcal{Y}$ is the support of $f$ (for simplicity we assume $\mathcal{Y}$ to be the same for all $\theta\in\Theta$).  In words, the approach of \citet{clarke1996} is to find the hypothetical dataset $\boldsymbol{y}^*$ that when combined with the baseline prior $\priorb$ produces the posterior distribution $q_{\priorb}(\cdot|\boldsymbol{y})$ with minimum KL-divergence from our true prior $\prior$. The effective prior sample size (EPSS from hereon) can then be quantified as the number of individual observations contained in $\boldsymbol{y}^*$.  Note that the density $f$ is a user specified hypothetical distribution for prior data, and is not necessarily the same as the model for any actual data.

The above approach is distinguished from most other methods (such as those mentioned below) in that it gives a specific dataset $\boldsymbol{y}^*$ which represents the information in the prior. An advantage of this approach is that $\boldsymbol{y}^*$ can potentially capture other aspects of the information contained in $\prior$ in addition to the EPSS. However,  $\boldsymbol{y}^*$ has no concrete relation to the likelihood or data at hand, which we consider to be a drawback, at least when question (ii) in the abstract is of primary interest.

\citet{Morita:2008} adopt a similar approach but measure distance using the difference of the second derivative of the log densities rather than KL divergence. Furthermore, to avoid the peculiarity of reporting a specific dataset $\boldsymbol{y}^{*}$, and to take account of uncertainty regarding the dataset, they take an expectation over $\boldsymbol{y}$, i.e., they compute $E[\frac{\partial^2}{\partial \theta^2}\log \postb(\cdot|Y)]$. This treatment of the hypothetical previous data $\boldsymbol{y}$ may be preferable to that of \citet{clarke1996}, but for the purpose of addressing (ii) the \citet{Morita:2008} method suffers from the same fundamental problem of not taking the likelihood of any actual data into account.

To address this limitation \citet{reimherr2014} introduced  the notion of prior-likelihood discordance and incorporated it in their measures of EPSS. The key change they proposed was to compare {\it two} posterior distributions rather than comparing a prior to a (hypothetical) posterior. To make the comparison, under each prior $\pi$, they consider  the expected mean squared error
when a draw from the posterior is used to estimate the true parameter $\theta_T$, i.e.,
\begin{align*}
U_{\prior,\theta_T}(I) = E_{\theta_T}[\mbox{MSE}(\prior,Y_I)] = \int_{\mathcal{Y}_I} \mbox{MSE}(\pi,\boldsymbol{y}_I) f(\boldsymbol{y}_I|\theta_T)d\boldsymbol{y}_I,
\end{align*}
where
\begin{align*}
\mbox{MSE}(\pi,\boldsymbol{y}_I) = \var_\prior(\theta|\boldsymbol{y}_I) + \left(E_\prior[\theta|\boldsymbol{y}_I]-\theta_T\right)^2.
\end{align*}
\citet{reimherr2014} use $I$ to indicate the information contained in the hypothetical data $\boldsymbol{y}_{I}$ and in their main examples  it represents the sample size (because the samples are assumed to be independent and identically distributed). Let $n$ be the sample size of the real data, denoted $\boldsymbol{y}\obs$. For hypothetical sample size $k\ll n$, \citet{reimherr2014} estimate the EPSS of an informative prior $\prior$ relative to a baseline prior $\priorb$ by the smallest integer $z$ such that
\begin{align*}
\hat{U}_{\prior,\hat{\theta}}(k) \approx \hat{U}_{\prior_b,\hat{\theta}}(k+z),
\end{align*}
where $\hat{\theta}$ is the maximum likelihood estimate of $\theta_T$ based on $\boldsymbol{y}\obs$, and $\hat{U}$ is computed by  averaging over datasets of size $k$ drawn from the empirical distribution (hence the constraint that $k\ll n$). By a slight abuse of terminology we refer to their averaging method as bootstrapping. One of the  novel aspects of this formulation is that $z$ is allowed to be negative. This is helpful when, for example, we are trying to assess if $\prior$ is a low-information prior and therefore might feasibly have a lower EPSS than $\pi^b$.

The approach of \citet{reimherr2014} described above has a number of advantages over earlier methods: (i) it focuses on the impact of the prior on posterior inference; (ii) it incorporates the likelihood, although for reduced data size;  and (iii) it proposes a potentially reasonable method for generating datasets to combine with $\prior$ and $\priorb$ (bootstrapping). There are however still some limitations of their approach. Firstly, their method averages over the data and therefore their measure of EPSS does not tell us what the impact of the prior is for the observed data $\boldsymbol{y}\obs$, which is of most interest in practice. Secondly, their approach relies on bootstrapping the data and estimating $\theta_T$ which both require $n$ to be large, but the impact of a prior is usually greatest and of most interest when $n$ is small. Lastly, their use of MSE is not necessarily the best way of quantifying the difference between two posterior distributions and therefore the prior impact. Indeed, there is in fact no reason to introduce the notion of a true parameter value $\theta_T$ in order to measure prior impact, a point we revisit in Section \ref{sec:gaussian}.  

\subsection{General formulations of prior sample size} \label{sec:general_form}

We now formulate general approaches for measuring EPSS that incorporate the key ideas in the previous work outlined in Section \ref{sec:review}  and in the the broader literature. Firstly, most approaches for measuring EPSS involve minimizing a distance or divergence between a probability density $q_\prior$ constructed using the user's prior and a probability density $q_{\priorb}$ constructed using a baseline prior. In some cases, the EPSS denoted $\pss$ is  directly computed based on this minimization. That is,
\begin{align}
\pss= h\left(\argmin_{(\boldsymbol{x},\boldsymbol{y})\in \mathcal{X}\times \mathcal{Y}}D(q_{\priorb}(\cdot|\boldsymbol{x}),q_\prior(\cdot|\boldsymbol{y}))\right),
\label{eqn:dist}
\end{align}
where $D$ is a distance or divergence between the probability densities  $q_\prior$ and $q_{\priorb}$, $\boldsymbol{x}$ and $\boldsymbol{y}$ denote data, and $h$ is a function of the optimal datasets $(\boldsymbol{x}^*,\boldsymbol{y}^*)$ which minimizes the $D$. This formulation is very general in that there are various choices for $h$ and $D$ and the data ($\boldsymbol{x}$ and $\boldsymbol{y}$) may be hypothetical, real, or re-sampled. In some cases only $\boldsymbol{x}$ is optimized, and $\mathcal{Y}$ is degenerate. For example, we may set $\mathcal{Y}=\{\boldsymbol{y}\obs\}$ (i.e., fix $q_\prior(\cdot|\boldsymbol{y})$ to be the posterior density $p_\prior(\cdot|\boldsymbol{y}\obs)$), or set $\mathcal{Y}=\emptyset$, where $\emptyset$ denotes the empty set (i.e., fix $q_\prior$ to be the prior $\prior$). The method of \citet{clarke1996} is an example of a EPSS measure that is defined according to (\ref{eqn:dist}).

Despite the flexibility of (\ref{eqn:dist}), from a statistical perspective it is not ideal to identify the EPSS $z$ with optimized datasets $\boldsymbol{x}^*$ and $\boldsymbol{y}^*$. If the datasets $\boldsymbol{x}$ are $\boldsymbol{y}$ are unknown then our uncertainty about them should be taken into account. 
Thus, most approaches to measuring EPSS do not directly consider (\ref{eqn:dist})  but rather an average quantity
\begin{align}
E[D(q_\prior(X),q_{\priorb}(Y))] = \int_{\mathcal{X}\times\mathcal{Y}}D(q_{\priorb}(\boldsymbol{x}),q_\prior(\boldsymbol{y})) f(\boldsymbol{x},\boldsymbol{y})d(\boldsymbol{x},\boldsymbol{y}),
\label{eqn:dist_average}
\end{align}
where $f$ denotes the joint density of $(\boldsymbol{x},\boldsymbol{y})$, and depends on the choice of $\mathcal{X}$ and $\mathcal{Y}$.
Optimization is used in selecting $\mathcal{X}\times\mathcal{Y}$ and it is the result of this final step that determines the EPSS. For example, we may have $\mathcal{X}\times\mathcal{Y}=  \mathbb{R}^m \times \emptyset$, where the $m$ minimizing (\ref{eqn:dist_average}) is to be determined, and $\pss=h(m)$. The approaches proposed by \citet{Morita:2008} and \citet{reimherr2014} are examples of methods that minimize (\ref{eqn:dist_average}).

Another work closely connected to (\ref{eqn:dist}) and (\ref{eqn:dist_average}) is that of
 \cite{ley2017}, 
which provides general upper and lower bounds for the 1-Wasserstein distance between the posterior distribution 
under a baseline prior and the posterior distribution under an informative prior. 
However,  \cite{ley2017} do not discuss EPSS directly, and their method  is limited to one-dimensional parameter spaces. 

\subsection{Observed effective prior sample size} In contrast to the measures discussed in Sections \ref{sec:review} and \ref{sec:general_form}, we propose that measures of EPSS should condition on the observed data at hand, {and in Section \ref{sec:definition} define the (posterior) mean observed prior effective sample size (MOPESS).} In our experience, applied statisticians are most interested in the impact of the prior distribution on the specific analysis that is being performed and reported, i.e., given (conditioning on) the observed data. Consequently, our proposed measures of EPSS (detailed below in Section \ref{sec:definition}) are closely related to the notion of {\it observed information}. 



Averaging is necessary when measuring the {\it self-information} of a random variable because the only uncertainty (and therefore potential for information) regards the value of the random variable. However, in the context of statistical inference, we are usually interested in observed mutual information, i.e., what can be learnt about the unknown parameter $\theta$ from the {\it observed} data. Therefore, expressing information as an average with respect to the data is often not the best option. For example, \citet{efron1978assessing} demonstrated that, in practice, observed Fisher information is often of greater relevance than expected Fisher information.

Thus, 
rather than computing (\ref{eqn:dist_average}),
our EPSS measures seek to capture the observed impact of a prior by conditioning on the observed data $\boldsymbol{y}\obs$, which from now on we usually denoted  $\boldsymbol{y}$ or $\{y_1,\dots,y_n\}$. In addition to $\boldsymbol{y}$,  we consider the hypothetical expanded dataset $\boldsymbol{x}=\boldsymbol{y}\cup\{x_{n+1},\dots,x_{n+r}\}$, where $r\in\mathbb{Z}_{\geq0}$ and $x_{n+1},\dots,x_{n+r}$ are future samples. It is still necessary to average over the unknown future samples, but this is done conditioning on the observed data and addresses real uncertainties regarding the future samples, as opposed to artificial uncertainties related to the observed data. 

\subsection{Definition of the observed prior sample size} \label{sec:definition}

Here we define EPSS from the perspective of observed information discussed above.  For a given $m=n+r$, the hypothetical expanded data set is   $\boldsymbol{ x}^{(m)}=\{x_1^{(m)},\dots,x_m^{(m)}\}$ $\equiv \boldsymbol{y}\cup\{x_{n+1}^{(m)},\dots,x_{m}^{(m)}\}$, where  the new superscripts indicate an association with the specific value of $m$, and will become useful in what follows. If $\boldsymbol{x}^{(m)}$ was known for all $m$ then intuitively we would want to choose the $m$ which  minimizes the distance between the two posterior distributions $q_{\pi}(\cdot|\boldsymbol{y})$ and $q_{\pi^b}(\cdot|\boldsymbol{x}^{(m)})$, where $\pi$ is our real prior whose EPSS is to be measured and $\pi^b$ is the baseline prior.  
Some care must be taken regarding the population from which $\boldsymbol{x}^{(m)}$
is assumed to originate. A researcher  who does not want to use $\pi$ would not  collect many additional samples and then attempt to find the $m$ to minimize the distance between their posterior and ours. Instead, they would simply collect a fixed number of additional samples $r=m-n$ (if they decided that more data were needed). 
Therefore the correct question to ask to quantify the EPSS of $\pi$ is as follows: if there are multiple independent researchers each of whom chooses a different value of $r$,  then whose inference 
will most closely agree with our inference? In other words, for the EPSS to correspond to normal scientific procedure, the hypothesized future samples must have some form of conditional independence across values of $m$ (to be specified shortly, see (\ref{eqn:ppd2}) below). 
This means that we cannot assume that $\boldsymbol{x}^{(m+1)}=\boldsymbol{x}^{(m)}\cup\{x_{m+1}\}$, so we write the full collection of expanded datasets as $\boldsymbol{x}_{\text{\tiny L}}$ $=\{\boldsymbol{y},\boldsymbol{y} \cup \{x_{n+1}^{(n+1)}\},\boldsymbol{y} \cup\{x_{n+1}^{(n+2)},x_{n+2}^{(n+2)}\},\dots,\boldsymbol{y} \cup\{x_{n+1}^{(L)},\dots,x_L^{(L)}\}\}$ $= \{\boldsymbol{x}^{(n)},\boldsymbol{x}^{(n+1)},\dots,\boldsymbol{x}^{(L)}\}$, where $L$ is the maximum feasible value of $m$, or in other words $L-n$ is the maximum feasible magnitude of the EPSS associated with $\pi$. In summary, for a given realization of $\boldsymbol{x}_{\text{\tiny L}}$, the EPSS essentially corresponds to the $m$ such that  $\boldsymbol{x}^{(m)}\in \boldsymbol{x}_{\text{\tiny L}}$ minimizes the distance between $q_{\pi}(\cdot|\bf{y})$ and $q_{\pi^b}(\cdot|\boldsymbol{x}^{(m)})$. 
To formally complete this specification 
a number of further details are needed, which we now discuss.  

Firstly, there are various measures of discrepancy  between two probability distributions  which could be used under the general framework given in Section~\ref{sec:general_form}. For example, Kullback-Leibler (KL) divergence was adopted by~\citet{clarke1996} to quantify prior-posterior discrepancy, and \citet{reimherr2014} used mean squared error as a discrepancy measure. There are also a number of other options such as   more general $f$-divergences~\citep{ali1966general,sason2016f}, of which KL divergence is a special case.  In this paper, we adopt the Wasserstein distance, for reasons to be explained shortly, but our methodology is  general and can use any discrepancy measure. For $p\geq 1$, let $\mu$ and $\nu$ be probability measures defined on $\mathcal{M}$ with finite $p^{\rm th}$ moment. The $p${\it-}Wasserstein distance between $\mu$ and $\nu$ is defined as
\begin{equation}
W_p(\mu, \nu) = \left(\inf _{{\gamma \in \Gamma (\mu ,\nu )}}\int _{{\mathcal{M} \times
\mathcal{M}}}d(x,y)^{{p}}\,{\mathrm  {d}}\gamma (x,y)\right)^{{1/p}},
\end{equation}
where $\Gamma(\mu,\nu)$ denotes  measures on $\mathcal{M}\times\mathcal{M}$ with marginals $\mu$ and $\nu$ respectively, and $d$ is a metric on $\mathcal{M}$. In the case of multivariate Gaussian distributions the $2${\it-}Wasserstein distance can be computed in closed form and more generally there are efficient software packages for approximating it given posterior samples, e.g., \citet{wasser}. The Wasserstein distance is widely used in statistics, e.g., in theoretical studies of Bayesian asymptotics (e.g. \citet{nguyen2016borrowing}), scalable Bayesian inference (e.g. \citet{srivastava2015wasp,minsker2017robust}) and variational inference (e.g.~\citet{ambrogioni2018wasserstein}).  
{It  enjoys a variety of desirable properties, making it a good choice for measuring the distance between the posterior distributions considered in our EPSS measures, as we now explain.} Intuitively, the Wasserstein distance captures the amount of ``effort'' needed to transform one probability distribution to another probability distribution, if we imagine the two probability densities as two piles of sands. From the prior influence perspective, this makes Wasserstein distance an appealing measure of  the closeness of two posterior distributions: we seek to quantify the amount of extra ``effort'' (in terms of extra samples) that is needed to transform the baseline prior posterior distribution into the posterior distribution under our prior $\pi$. In the current context a successful transformation reproduces  the posterior distribution under $\pi$ in terms of variance, location, and tail behavior, which are all criteria  well measured  by the Wasserstein metric. For example, in \citet{ley2017}, the Wasserstein distance is the adopted metric for assessing prior influence on Bayesian inference. 
In Appendix~\ref{sec:app_distance}, we demonstrate our method using an alternative discrepancy measure (namely KL divergence) in a Gaussian conjugate model example, and thereby further illustrate both the generality of our approach by giving sensible results of EPSS based on the alternative measure and point out the connections between the two discrepancy measures. 

Secondly, we must allow for the possibility  that our prior $\pi$ is in fact {\it less} informative than the baseline prior $\pi^b$. In our previous discussions, we take for granted that our 
prior $\pi$ contains more ``information'' than the baseline prior and therefore that the baseline prior should be supplemented by extra samples. {However, in practice, the prior $\pi$ could potentially have less impact on the analysis than the baseline prior $\pi^b$. This happens when the prior $\pi$ is more diffuse than the baseline $\pi^b$ or is similarly diffuse but has  greater location agreement with the data  than $\pi^b$.} 
Thus, in addition to combining extra samples with $\pi^b$, it is natural to also consider the alternative of combining extra samples with $\pi$ and finding the minimum of the distance $W_2\left(q_{\pi^b}(\cdot|\boldsymbol{y}), q_{\pi}(\cdot|\tilde{\boldsymbol{x}}^{(m)})\right)$ across values of $m$. Here $\tilde{\boldsymbol{x}}_m \in \tilde{\boldsymbol{x}}_{\text{\tiny L}}$, where $\tilde{\boldsymbol{x}}_{\text{\tiny L}}=\{\tilde{\boldsymbol{x}}^{(n)},\tilde{\boldsymbol{x}}^{(n+1)},\dots,\tilde{\boldsymbol{x}}^{(L)}\}$ is a second realization of $\boldsymbol{x}_{\text{\tiny L}}$, which we  assume is independent, see (\ref{eqn:ppd2}). 
We avoid assuming that $\tilde{\boldsymbol{x}}_{\text{\tiny L}}=\boldsymbol{x}_{\text{\tiny L}}$ for essentially the same reason that we do not assume $\boldsymbol{x}^{(m+1)}=\boldsymbol{x}^{(m)}\cup\{x_{m+1}\}$: for the interpretation of the EPSS to correspond to normal scientific procedure, we cannot assume that each individual  researcher computes both $q_{\pi^b}(\cdot|\boldsymbol{x}^{(m)})$ and  $q_{\pi}(\cdot|\boldsymbol{x}^{(m)})$ and then decides which to use depending on whether $W_2\left(q_{\pi^b}(\cdot|\boldsymbol{x}^{(m)}), q_{\pi}(\cdot|\boldsymbol{y})\right)$ or $W_2\left(q_{\pi^b}(\cdot|\boldsymbol{y}), q_{\pi}(\cdot|\boldsymbol{x}^{(m)})\right)$ is smaller. 
We instead assume there are two researchers in the population for each value of $m$, one who computes  $q_{\pi^b}(\cdot|\boldsymbol{x}^{(m)})$ and one who computes $q_{\pi}(\cdot|\tilde{\boldsymbol{x}}^{(m)})$, and since the reseaerchers will likely have different laboratories it is natural to assume that  $ \boldsymbol{x}^{(m)}$ and $\tilde{\boldsymbol{x}}^{(m)}$ are independent.
The importance of this point is mainly conceptual:
setting $\tilde{\boldsymbol{x}}_{\text{\tiny L}}=\boldsymbol{x}_{\text{\tiny L}}$ did not substantially change our results compared with allowing $\tilde{\boldsymbol{x}}_{\text{\tiny L}}$ and $ \boldsymbol{x}_{\text{\tiny L}}$ to be independent. 
We write 
$\boldsymbol{x}_{\text{\tiny L}}^{\text{all}}=\boldsymbol{x}_{\text{\tiny L}} \cup \tilde{\boldsymbol{x}}_{\text{\tiny L}}= \{\boldsymbol{x}^{(n)},\tilde{\boldsymbol{x}}^{(n)},\boldsymbol{x}^{(n+1)},\tilde{\boldsymbol{x}}^{(n+1)},\dots,\boldsymbol{x}^{(L)},\tilde{\boldsymbol{x}}^{(L)}\}$ to denote all the future samples combined, and for conciseness introduce the notation $\Wm$ and $\Wmflip$ to denote the distances $W_2\left( q_{\pi^b}(\cdot|\boldsymbol{x}^{(m)}),q_{\pi}(\cdot|\boldsymbol{y})\right)$ and $W_2\left(q_{\pi^b}(\cdot|\boldsymbol{y}), q_{\pi}(\cdot|\tilde{\boldsymbol{x}}^{(m)})\right)$, respectively.

Lastly, we define the sign function for our EPSS measure, which captures whether the prior has a greater or smaller influence on the inference than $\pi^b$:
\begin{align*}
    S_{n}(\boldsymbol{x}_{\text{\tiny L}}^{\text{all}}) = \left\{\begin{array}{cc}
       1  & \text{if} \quad \min_{n \leq m \leq L}\left\{\Wm\right\} \leq \min_{n \leq m \leq L}\left\{\Wmflip\right\} \\
        -1 & \text{if} \quad \min_{n \leq m \leq L}\left\{\Wm\right\} > \min_{n \leq m \leq L}\left\{\Wmflip\right\}.
    \end{array}\right.
\end{align*}
The sign function identifies to which prior extra samples must be added in order to reduce the discrepancy between the two posteriors. Our EPSS measures are relative to the baseline prior $\pi^b$, which can therefore be assumed to have EPSS zero. Thus, if the minimum distance is achieved by adding  extra samples to the baseline prior $\pi^b$, that indicates that our prior $\pi$ has a greater impact on the inference than $\pi^b$, and therefore should have positive EPSS; otherwise, it should have negative EPSS.

The underlying quantity of interest, namely
the EPSS value for a specific realization of $\boldsymbol{x}_{\text{\tiny L}}^{\text{all}}$, can now be explicitly defined:
\begin{align}
M_n(\boldsymbol{x}_{\text{\tiny L}}^{\text{all}}) = \left\{\begin{array}{cc}
  {\rm argmin}_{n \leq m \leq L} \left\{\Wm\right\} - n\quad {\text if}\quad S_{n}(\boldsymbol{x}_\text{\tiny L}^{\text{all}})= 1\\
   n - {\rm argmin}_{n \leq m \leq L} \left\{\Wmflip\right\}\quad {\text if}\quad S_{n}(\boldsymbol{x}_\text{\tiny L}^{\text{all}})= -1.
\end{array}
    \right.\label{eqn:meaningfulM}
\end{align}
If $\min_{m\geq n}\left\{\Wm\right\} > \min_{m\geq n}\left\{\Wmflip\right\}$ (and thus $S_n(\boldsymbol{x}_{\text{\tiny L}}^{\text{all}}) = - 1$) then this suggests that $\pi$ is less informative than $\pi^b$, which is why $M_n$ is defined to be negative in this case, as explained above. An alternative strategy for defining negative  EPSS is to allow {\it fewer} than $n$ samples to be combined with $\pi^b$, i.e., to remove some of the observed data when computing the posterior distribution under $\pi^b$. However, we found this to have both conceptual and practical disadvantages compared with the above definition, e.g., if samples are removed from the observed data then the resulting EPSS measure is highly sensitive to the order in which the samples were collected (or the order has to be averaged over which introduces additional challenges and computation). Furthermore, the strategy of removing samples does not fully use the information in the observed dataset and thus causes more variability in the final estimates of the EPSS.

In practice we do not know the values of the future samples contained in $\boldsymbol{x}_{\text{\tiny L}}^{\text{all}}$, and our uncertainty about them 
is naturally captured by the posterior predictive distribution computed under our real prior $\pi$,
\begin{align}
p(\boldsymbol{x}_{\text{\tiny L}}^{\text{all}}|{\boldsymbol{y}}, \pi) &= \int_\Theta p(\boldsymbol{x}_{\text{\tiny L}}|{\boldsymbol{y}}, \theta)p(\tilde{\boldsymbol{x}}_{\text{\tiny L}}|{\boldsymbol{y}}, \theta)q_\pi(\theta| {\boldsymbol{y}})d\theta \label{eqn:ppd1}\\&= \int_\Theta \prod_{m=n}^L p(\boldsymbol{x}^{(m)}|{\boldsymbol{y}}, \theta)\prod_{m=n}^L p(\tilde{\boldsymbol{x}}^{(m)}|{\boldsymbol{y}}, \theta)q_\pi(\theta| {\boldsymbol{y}})d\theta,
\label{eqn:ppd2}
\end{align}
where for the reasons discussed above we have assumed that $\boldsymbol{x}_{\text{\tiny L}}$ is conditionally independent of $\tilde{\boldsymbol{x}}_{\text{\tiny L}}$, and the future samples $\boldsymbol{x}^{(m)}$ (and $\tilde{\boldsymbol{x}}^{(m)}$) are conditionally independent across values of $m$, given $\boldsymbol{y}$ and $\theta$.  
Unconditionally, all the future samples are dependent, which corresponds to the real-world in that all additional samples collected  
would 
be generated using the same true (but unknown) value of $\theta$. 
The posterior predictive distribution (\ref{eqn:ppd1})-(\ref{eqn:ppd2}) in turn induces the
 posterior distribution of $M_n(\boldsymbol{x}_{\text{\tiny L}}^{\text{all}})$, denoted $\mathcal{F}_{\rm EPSS}$, which forms a complete summary of the EPSS of $\pi$.  
To provide a single univariate measure of EPSS we suggest reporting the posterior mean of $M_n$, denoted $\overline{M}_n$, which is the Bayes estimate of $M_n$ under squared error loss. In practice, it may be helpful to additionally report several quantiles of  $\mathcal{F}_{\rm EPSS}$. 
{In the remainder of the paper, we use the acronym OPESS (observed prior effective sample size) to refer to an individual realization of $M_n(\boldsymbol{x}_L^{\rm all})$, i.e.,  corresponding to a specific realization of $\boldsymbol{x}_L^{\rm all}$,   and MOPESS (mean OPESS) to refer to the posterior mean estimate $\overline{M}_n = \sum_{j=1}^SM_n^{(j)}$, where $S$ denotes the number of simulated realizations of $M_n$.}

In the above, the independence of $\boldsymbol{x}^{(m)}$ (and $\tilde{\boldsymbol{x}}^{(m)}$) across different values of $m$ has the  statistical advantage that the posterior distribution of  $M_n$ has relatively low variance (for a given observed dataset $\boldsymbol{y}$). This increases the practical appeal of $M_n$  and also means that the computational cost of drawing $\boldsymbol{x}_{\text{\tiny L}}^{\text{all}}$ is somewhat offset because 
posterior summaries of $M_n$ (e.g., the posterior mean)  can be estimated with relatively few Monte Carlo simulations. Further note that the computation can easily be parallelized.

Algorithm~\ref{algo:OPESS_compute} summarizes our general procedure for computing the posterior mean observed prior effective sample size (MOPESS). The procedure is widely applicable and can be implemented for a large family of models beyond the specific cases considered in this paper. 
Naturally, we use analytical forms of the posterior distributions and the Wasserstein distances when available; otherwise, we choose from several approximation strategies. We acknowledge that when the approximations are inaccurate, the resulting  MOPESS estimates  can be substantially affected, and therefore in practice it is important  to check  the effectiveness of each approximation. For example, in Step 2(b), if we use importance sampling or the Markov chain Monte Carlo (MCMC) algorithms~\citep{marin2007bayesian,liu2008monte,brooks2011handbook} to approximate the posterior distributions, we must  check the effective sample size (and other diagnostics), see e.g. \citet{geweke1991evaluating,gelman1992inference,kass1998markov,mengersen1999mcmc,yang2018parallelizable} for more details, to make sure that the samples well represent the posteriors. 

\begin{algorithm}[t]
\SetAlgoLined
\vspace{0.1cm}
{\bf Step 1:} 
Compute posterior distributions $q_{\pi^b}(\cdot|\boldsymbol{y})$ and $q_{\pi}(\cdot|\boldsymbol{y})$ either analytically 
or numerically (posterior samples). \\ 
\vspace{0.2cm}
 {\bf Step 2:} Repeat the following for $j=1,\dots S$:\\\vspace{0.2cm}
\begin{minipage}{0.05\textwidth}
 \quad
\end{minipage}\begin{minipage}{0.9\textwidth}
   {\bf Part a:} Generate extra samples $\boldsymbol{x}_L^{\rm all}=\boldsymbol{x}_L\cup \boldsymbol{\tilde{x}}_L$ from  posterior predictive distribution  under $\pi$, see (\ref{eqn:ppd1})-(\ref{eqn:ppd2}).\vspace{0.1cm}
   
{\bf Part b:} 
For $m=n+1,\dots,L$, compute the posterior distributions $q_{\pi^b}(\cdot|\boldsymbol{x}^{(m)})$ and $q_{\pi}(\cdot|\tilde{\boldsymbol{x}}^{(m)})$ either analytically or via one of the following: 
\begin{itemize}
    \item Importance sampling: use either $q_{\pi^b}(\cdot|\boldsymbol{x}^{(m^*)})$ or $q_{\pi}(\cdot|\tilde{\boldsymbol{x}}^{(m^*)})$ as importance functions, for $n \leq m^*< m$. 
    \item Appropriate Markov chain Monte Carlo (MCMC) algorithm.
\end{itemize}
{\bf Part c:} Compute the Wasserstein distances $\{W_2(m), \widetilde{W}_2(m), n\leq m\leq L\}$ analytically or via one of the following:
\begin{itemize}
    \item Gaussian approximation to the posteriors and use of  the analytical Wasserstein distance between Gaussian distributions. 
    \item Numerical approximation based on posterior samples obtained in Step 3, e.g., using the \texttt{R} package by~\citet{wasser}. 
\end{itemize}
{\bf Part d:} Calculate the OPESS $M_n^{(j)}$ given by (\ref{eqn:meaningfulM}).
\end{minipage}\begin{minipage}{0.05\textwidth}
 \quad
\end{minipage} \vspace{0.1cm}

 {\bf Step 3:} Report the MOPESS: $\overline{M}_n = \frac{1}{S} \sum_{j=1}^SM_n^{(j)}$. \\\vspace{0.1cm}
  \caption{General procedure for computing MOPESS. \label{algo:OPESS_compute}}

\end{algorithm}

\section{Gaussian illustration}\label{sec:illustration}


\subsection{Setup} \label{sec:gaussian}

Let $\{y_i,1\leq i\leq n\}$ 
be independent observations from a Gaussian distribution with mean $\mu$ and variance $\sigma^2$, i.e., $y_i \stackrel{\rm i.i.d.}{\sim} \mathcal{N}(\mu, \sigma^2)$, for $i=1,\dots,n$. Assume that $\sigma^2$ is known, our prior for $\mu$ is a conjugate prior, denoted $\prior^{c} (\mu)\equiv \mathcal{N}(\mu_0, \lambda_0^2)$, and the baseline prior is $\prior^f(\mu)\propto 1$. Regarding the expanded dataset, as before we have $x^{(m)}_i=y_i$ for $i=1,\dots,n$, and suppose that, if collected, the hypothetical future samples would be drawn from the same distribution as the observations, i.e., that {\it hypothetically} $x^{(m)}_i \stackrel{\rm i.i.d.}{\sim} \mathcal{N}(\mu, \sigma^2)$, for $i=n+1,\dots,m$.  Let $\pi_n^A$ and $\pi_m^A$ denote the posterior distribution obtained by combining the prior $\pi^A$ with the observed and expanded dataset, respectively, for $A=c,f$. Lemma~\ref{lemma:gaussianposterior} specifies these posterior distributions and the 2-Wasserstein distance between $\pi^f_m$ and $\pi^c_n$ and between $\pi^f_n$ and $\pi^c_m$. The proof is straightforward and is omitted. 

\begin{lemma}
Suppose that $m\geq n$. 
For $u=n,m$, we have
\begin{equation}
\prior_u^f = \mathcal{N}\left(\overline{x}_u, \frac{\sigma^2}{u}\right), \quad \prior_u^c = \mathcal{N}\left(\mu_u=(1-w_u) \mu_0 + w_u \overline{x}_u, \frac{\sigma^2}{u + z}\right),
\label{eqn:gaussianposteriors}
\end{equation}
where $w_u = {u}/(u+z)$, $z = \sigma^2/\lambda_0^2$, and
\begin{equation*}
\quad \overline{x}_u = \frac{1}{u} \sum_{i=1}^u x_i^{(u)},
\end{equation*}
Furthermore, the $2$-Wasserstein distance between $\pi_m^f$ and $\pi_n^c$, and between $\pi_n^f$ and $\pi_m^c$, is
\begin{align}
\Wm \equiv W_2\left( \pi_m^f,\pi_n^c \right)
&= D_{m, n} +  \left(\frac{\sigma}{\sqrt{m}} - \frac{\sigma}{\sqrt{n+z}}, \right)^2,\label{eq:W2gaussian1}\\
\Wmflip \equiv W_2\left(\pi_m^c, \pi_n^f \right)
&= D_{n, m} +  \left(\frac{\sigma}{\sqrt{n}} - \frac{\sigma}{\sqrt{m+z}}, \right)^2,\label{eq:W2gaussian2}
\end{align}
respectively, where $D_{u,v} = (\bar{x}_u - \mu_v)^2$, for $u,v=n,m$.
\label{lemma:gaussianposterior}
\end{lemma}

If the values of $x_{n+1}^{(m)}\dots,x_m^{(m)}$ were known for all $m \in \{n,\dots, L\}$, then we would know (\ref{eq:W2gaussian1}) and (\ref{eq:W2gaussian2}) and hence $M_n(\boldsymbol{x}_{\text{\tiny L}}^{\text{all}})$. Since in practice the future samples are unknown, our method introduced in Section  \ref{sec:definition} specifies that we should look at the  posterior distribution of $M_n(\boldsymbol{x}_{\text{\tiny L}}^{\text{all}})$, the observed prior effective sample size (OPESS). In the current scenario, realizations of the OPESS can be generated by drawing $\mu^*\sim\pi_n^c$ and then drawing $\boldsymbol{x}_{\text{\tiny L}}^{\text{all}}$ from \begin{align*}
    p(\boldsymbol{x}_{\text{\tiny L}}^{\text{all}}|\boldsymbol{y},\mu^*)=\prod_{m=n+1}^L\prod_{i=n+1}^m \mathcal{N}(x_i^{(m)}|\mu^*,\sigma^2)\prod_{i=n+1}^m \mathcal{N}(\tilde{x}_i^{(m)}|\mu^*,\sigma^2).
\end{align*}


\subsection{Numerical Results}
\label{subsec:gaussian_numerical}


\begin{figure}[tbph]
\includegraphics[width=0.49\textwidth,trim=0mm 132mm 0mm 0mm,clip]{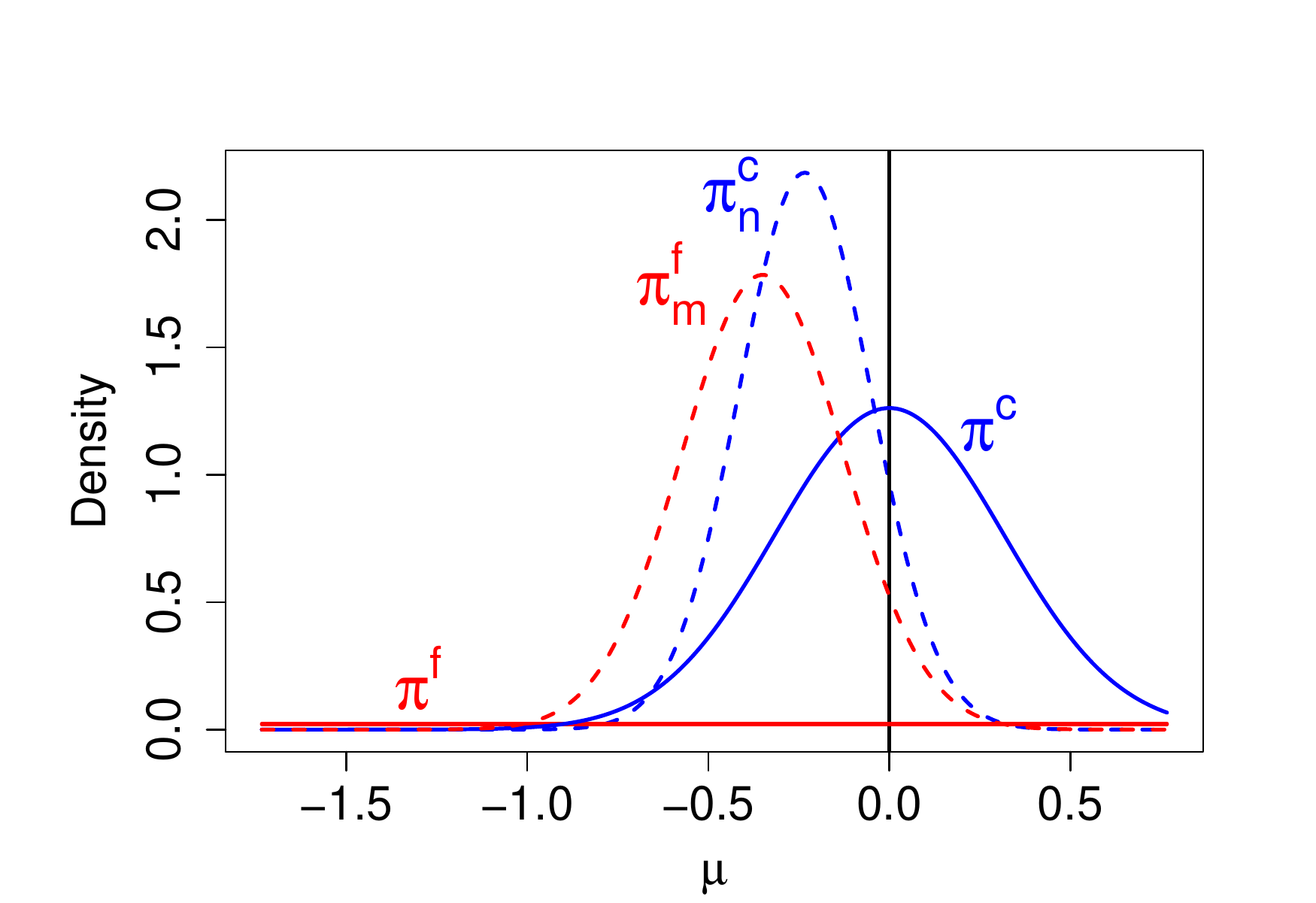}\includegraphics[width=0.49\textwidth,trim=0mm 10mm 0mm 130mm,clip]{presentation_plot_prior0}
\caption{(Left)  posteriors $\pi_n^c$ and $\pi_n^f$ (dashed lines) and priors $\pi^c$ and $\pi^f$ (solid lines) for a single simulated dataset $\boldsymbol{y}$. (Right) distribution of the MOPESS ($\overline{M}_n$) across 300 simulated datasets. The vertical line shows the nominal EPSS of $10$.
\label{fig:prior0}}
\end{figure}

Suppose that $\mu=\mu_0=0$, $\sigma^2=1$, and $\lambda_0^2=0.1$. Under these settings, the nominal sample size of the informative prior $\prior^{c}$ is $10$ because of the following three information based analogies between the prior and data: (i) if $n=10$ then the Fisher information is $n/\sigma^2=1/\lambda_0^2=10$, (ii) if $n=10$ then $\bar{y}_n \sim \prior^{c}$, and (iii) for any $n$, the posterior distribution is $\mathcal{N}(0,\sigma^2/(n+10))$. 

The left panel of Figure \ref{fig:prior0} shows the posteriors $\pi_n^c$ and $\pi_m^f$ (dashed lines) for a single example observed value of $\bar{y}_n$, with $n=20$. The priors $\pi^c$ and $\pi^f$ are also plotted (solid lines). The right panel of Figure \ref{fig:prior0} shows the distribution of the MOPESS, i.e., of $\overline{M}_n$,  across 300 datasets. Interestingly, in the current context $\overline{M}_n$ is quite variable and is always higher than the nominal EPSS of $10$ (vertical line).

The top left panel of Figure \ref{fig:prior0_cases} shows that the variation in $\overline{M}_n$ is due to variation in $\bar{y}_n$ across the 300 datasets, because $\bar{y}_n$ (and $n$) determines the posterior distribution $\pi_n^c$. The bold line in the left panel of  Figure \ref{fig:prior0_cases} is a LOESS (local polynomial regression) fit to the plotted points. For each dataset, the value of $\overline{M}_n$ was computed via 10,000 Monte Carlo samples of $\boldsymbol{x}_{\text{\tiny L}}^{\text{all}}$ and for a given $\bar{y}_n$ the scatter is entirely due to Monte Carlo error, i.e., with enough Monte Carlo samples all the points would lie exactly on a curve similar to the bold line plotted. The top right panel of Figure \ref{fig:prior0_cases} re-plots the LOESS line from the left panel and shows three quantiles of ${M}_n$: the median (black short dash curve), 95\% quantile (green long dash curve), and 5\% quantile (blue dash-dot curve). 
The plot illustrates that even for a fixed value of $\bar{y}_n$ the value of $M_n(\boldsymbol{x}_{\text{\tiny L}}^{\text{all}})$ can be highly variable across Monte Carlo realizations of $\boldsymbol{x}_{\text{\tiny L}}^{\text{all}}$. 

We now turn to the bottom panels of Figure \ref{fig:prior0_cases} which help to explain the phenomena seen in the top panels. The bottom left panel of Figure \ref{fig:prior0_cases} corresponds to the dataset indicated by a ``+" symbol in the top left plot, i.e., the case where $\bar{y}_n$ is furthest from $\mu=\mu_0=0$ across all 300 simulated datasets.  The resulting posteriors $\pi_n^c$ and $\pi_n^f$ are plotted in the bottom left panel and $\pi_n^c$ is seen to be pulled towards zero by $\prior^c$. This scenario corresponds to the greatest Wasserstein distance between the posteriors $\pi_n^c$ and $\pi_n^f$ because  the difference in posterior means is given by $(1-w_n)\bar{y}_n$, where $w_n=(n/\sigma^2)/((1/\lambda_0^2)+(n/\sigma^2))=2/3$ (the posterior variances only depend on $n$).  The top left panel shows that in this case the MOPESS is larger than for the other simulations and in particular is around $14.5$,  which is considerably larger than the nominal EPSS of $10$.

The bottom right panel of Figure \ref{fig:prior0_cases} illustrates the case where $\bar{y}_n$ is closest to $\mu=\mu_0=0$ across the 300 simulated datasets, i.e., the dataset indicated by a cross in the top left panel.  From the bottom right panel we can see that for this dataset both posteriors are centered at zero. Specifically, the posteriors have substantial overlap because the conjugate prior $\pi^c$ is centered at $\mu=\mu_0=0\approx \bar{y}_n$ and so does not cause $\pi_n^c$ to have a substantially different mean to $\pi_n^f$, only a smaller variance. Returning to the top left panel we can see that this case corresponds to a MOPESS value of around $10.5$, which is one of the smallest across  our simulations. In conclusion, we can see that the MOPESS is larger the further $\bar{y}_n$ is from zero and that this is because the prior has more impact on the posterior in these cases. Thus, at least in this simple example,  our MOPESS measure of EPSS seems to have an intuitive interpretation that well captures the way the prior impact changes with the observed data.

\begin{figure}[t]
\includegraphics[width=0.49\textwidth,trim=0mm 10mm 0mm 0mm,clip]{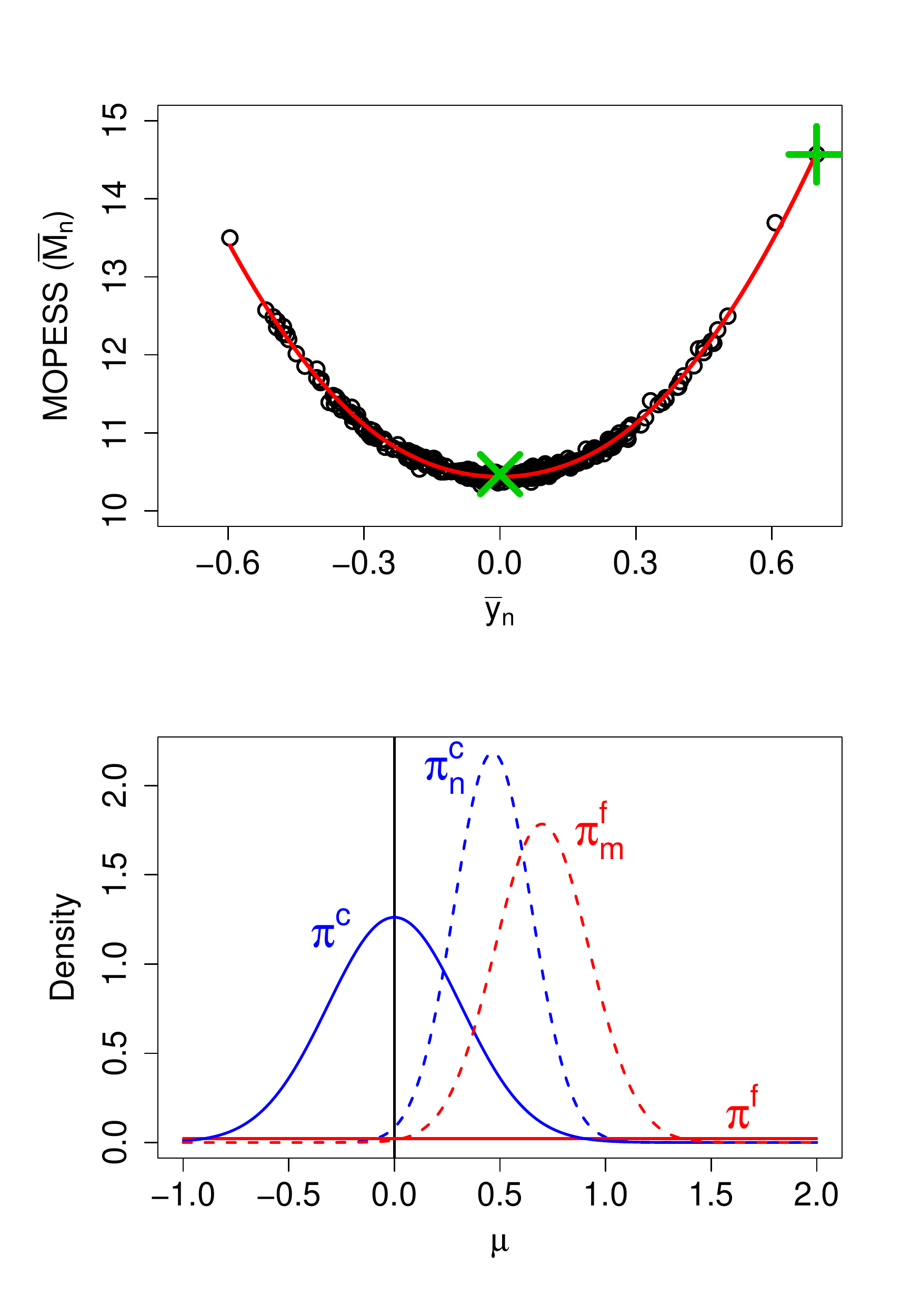}
\includegraphics[width=0.49\textwidth,trim=0mm 10mm 0mm 0mm,clip]{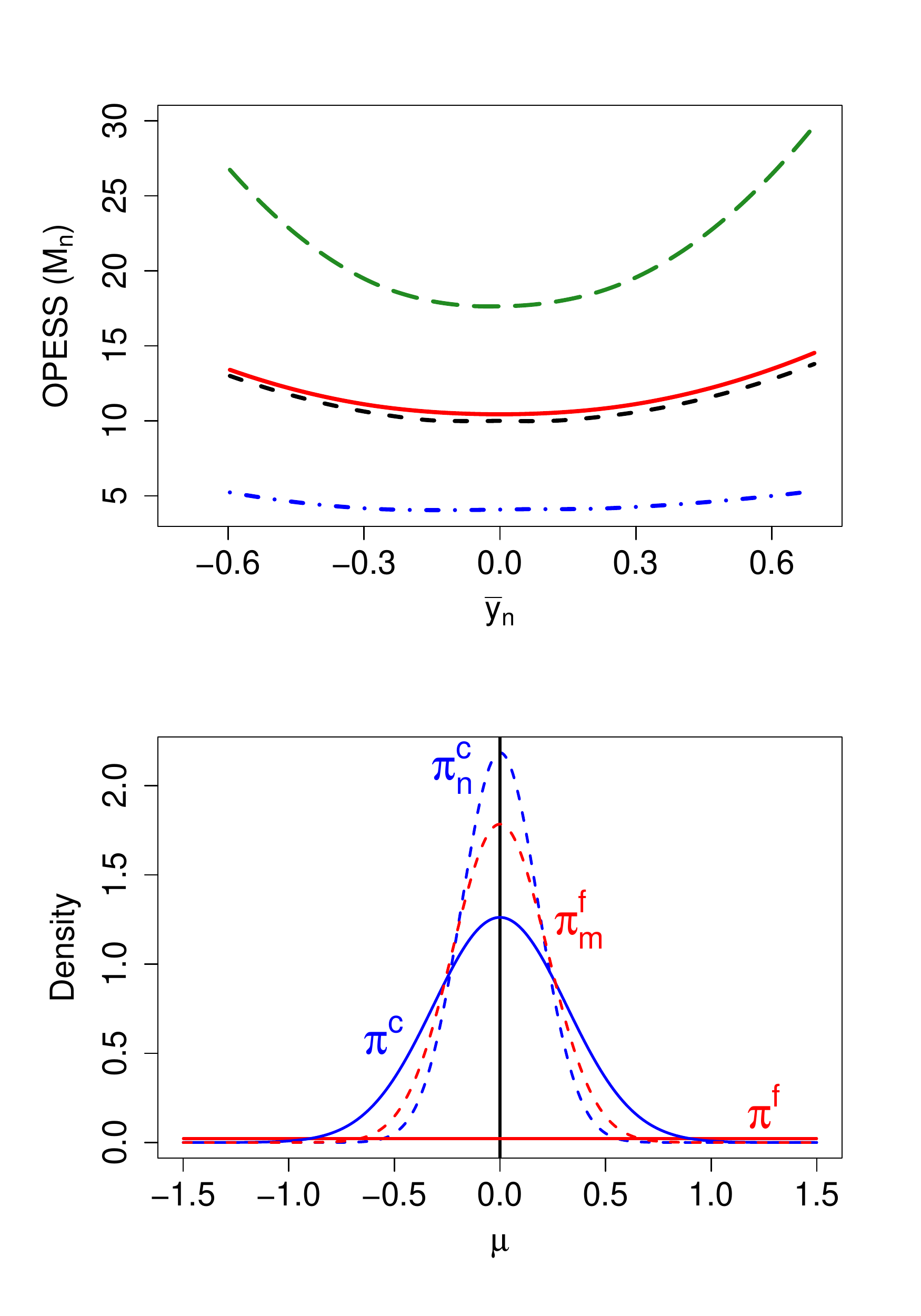}
\caption{(Left top)  MOPESS ($\overline{M}_n$) as a function of the observed data mean $\bar{y}_n$. Each point corresponds to one of the 300 simulated datasets, the bold curve shows a LOESS (local polynomial regression) fit to the points. The green ``+" symbol and cross  indicate the datasets for which $|\bar{y}_n-\mu_0|$ is greatest and smallest, respectively. 
(Right top)  quantiles of the posterior distribution of $M_n$ as a function of $\bar{y}_n$ including the 95\% quantile (long dash curve), median (short dash curve),  and 5\% quantile (dash-dot curve). The solid line is the same as in the left panel. (Left bottom) posteriors $\pi_n^c$ and $\pi_n^f$ (dashed lines) and priors $\pi^c$ and $\pi^f$ (solid lines) for the dataset indicated by a ``+" symbol in the top left plot. (Bottom right) the same as in the bottom left panel except for the the dataset indicated by a cross in the top left plot.
\label{fig:prior0_cases}}
\end{figure}

Turning our attention from the {\it shape} of the curve in the top left panel of Figure \ref{fig:prior0_cases} to the specific values, we note that $\overline{M}_n$ is always greater than the nominal EPSS of $10$. In Section \ref{sec:sampling_dist} we illustrate that there is a good explanation for this  {\it location} discrepancy: classical information measures consider the prior in isolation and only correspond to the prior impact if there is no data. As soon as some data are collected there is, on average, always some level of disagreement between the prior and the data and therefore  $\overline{M}_n$ is usually greater than the nominal EPSS, at least in the current Gaussian conjugate model example. However, in specific circumstances it is possible for the MOPESS $\overline{M}_n$ to be less than the nominal EPSS $z$, namely, if the location discrepancy between the initial posteriors $\pi^c_n$ and $\pi^f_n$  is small relative to $z$, a scenario which is discussed further in Section \ref{subsec:beta-binomial} and Appendix \ref{app:small_epss}. 
Furthermore, in the right panel of Figure \ref{fig:prior0_cases}, the 5\% quantile  of $M_n(\boldsymbol{x}_{\text{\tiny L}}^{\text{all}})$ (dash-dot curve)  shows that there often exist {\it some} realizations of $\boldsymbol{x}_{\text{\tiny L}}^{\text{all}}$ such that the value of $M_n$ is less than the nominal EPSS. Indeed, $\pi_m^f$ may by chance be close to $\pi_n^c$ after $m-n < 10$ additional samples. More generally, the uncertainty represented by the quantiles in the top right panel of Figure \ref{fig:prior0_cases} corresponds to real-world uncertainty about future samples and in particular  how many will be needed to obtain comparable inference to that provided by $\pi_n^c$. 

The bottom right panel of Figure \ref{fig:prior0_cases} discussed above corresponds to the case of a {\it super-informative} prior mentioned by \citet{reimherr2014}. In their paper, {\it super-information} refers to the case where the mean of the prior $\prior^c$ is closer (in terms of squared difference) to the true mean of the data $\mu$ than expected based on the variance of the prior. For example, in the bottom right panel of Figure \ref{fig:prior0_cases}, the prior mean is closer than would be expected if the prior had been constructed by computing the posterior distribution based on $10$ earlier observations. The measures proposed by \citet{reimherr2014} give an EPSS larger than the nominal EPSS of $10$ in this super-informative prior context because the prior is unexpectedly accurate (in some other cases their measure is lower than the nominal value). However, as we have seen, the MOPESS is  relatively low in the case where $\mu_0\approx \bar{y}_n$  (although still larger  than the nominal EPSS of $10$). Thus, we re-interpret the super-information phenomenon as a low-impact phenomenon. 
Indeed, if the majority of the prior mass is unexpectedly close to the data mean then the prior has less impact than expected.
A conceptual difference between  super-information and low-impact is that in  the case of the latter the true value of $\mu$ is irrelevant because, once we condition on the observed data $\boldsymbol{y}$, the true parameter value $\mu$ does not tell us anything about the {\it impact} of the prior on  inference. 


This section has illustrated  the limitation of only reporting the nominal EPSS of $10$, namely,  the actual impact of the prior 
depends on the observed data. The promise of our approach is that it fully takes the observed data into account, and in this respect it is unique among the measures of EPSS that have been proposed to the best of our knowledge.

\section{Method justification and theory}\label{sec:justification_theory}

\subsection{Illustration regarding choice of sampling distribution for extra observations}
\label{sec:sampling_dist}

Let $r=m-n$ and denote the additional samples collected by $s_1^{(m)},\dots,s_r^{(m)}$, i.e., $\{x_1^{(m)},\dots,x_m^{(m)}\} = \{y_1,\dots,y_n,s_1^{(m)},\dots,s_r^{(m)}\}$. We write $\bar{s}_r$ to denote $\frac{1}{r}\sum_{i=1}^r s_i^{(n+r)}$. Returning to the Gaussian conjugate model introduced earlier, we have
\begin{align}
\Wm
&= \left(\mu_n -\frac{n}{m}\bar{y}_n-\frac{r}{m}\overline{s}_r\right)^2 +  \left(\frac{\sigma}{\sqrt{m}} - \frac{\sigma}{\sqrt{n+z}} \right)^2,\label{eq:illustrationw2a}\\
\Wmflip
&= \left(\bar{y}_n-\frac{n+z}{m+z}\mu_n - \frac{r}{m+z}\bar{s}_r \right)^2 +  \left(\frac{\sigma}{\sqrt{n}} - \frac{\sigma}{\sqrt{m+z}} \right)^2.\label{eq:illustrationw2b}
\end{align}
Recall that in our approach described in Section \ref{sec:definition} the future samples are drawn from the posterior predictive distribution (\ref{eqn:ppd1})-(\ref{eqn:ppd2}) under $\pi^c$, meaning that
\begin{align}
\bar{s}_r|\bar{y}_n \sim N\left(\mu_n,\left(\frac{1}{r}+\frac{1}{n+z}\right)\sigma^2\right).\label{eqn:sr}
\end{align}
 In contrast to our approach, \citet{Morita:2008} sample  from the distribution of hypothetical previous data and \citet{reimherr2014} bootstrap the observed data. Our proposed sampling method is therefore not the only option and in order to provide justification for our choice it is instructive to consider the behavior of $M_n$ under several sampling methods. 
 To investigate this Proposition \ref{THM:SAM_DIST} below considers the case where $\bar{s}_r$ is exactly equal to the mean of its distribution, denoted $\gamma$. If the behavior of $M_n$ for this value of $\bar{s}_r$ does not make sense then there is little hope that the corresponding sampling method is useful, and if it does make sense then the investigation may offer  valuable insights. The proof of Proposition \ref{THM:SAM_DIST} is given in Appendix \ref{app:proof_prop}.
 
 \begin{proposition}\label{THM:SAM_DIST}
Suppose that $\bar{s}_r=E[\bar{s}_r|\bar{y}_n]=\gamma$. Under this scenario we have the  following results:
\begin{enumerate}
    \item  (Posterior predictive sampling) If $\gamma=\mu_n$, then $M_n\geq z$.
    \item   (Bootstrap sampling) If $\gamma=\bar{y}_n$, then there exists $\epsilon_s$ and $\epsilon_l$, such that $M_n=z$ whenever $|\bar{y}_n -\mu_n|<\epsilon_s$, and  $M_n<0$ whenever $|\bar{y}_n -\mu_n|>\epsilon_l$.
    \item (Prior sampling) If $\gamma = \mu_0$, then $M_n=z$.
\end{enumerate}

\end{proposition}

Result (a) of Proposition \ref{THM:SAM_DIST} corresponds to our proposed method of sampling the future samples from the posterior predictive distribution (\ref{eqn:ppd1})-(\ref{eqn:ppd2}). 
The result is consistent with  the top right panel of Figure \ref{fig:prior0_cases} in Section \ref{subsec:gaussian_numerical}, which shows that the median (dashed curve) value of $M_n$ is always equal to or greater than $z$. 
To gain further intuition consider the distance $W_2(m)$ under the condition of  Proposition \ref{THM:SAM_DIST}: 
\begin{align}
    \Wm = \left(\frac{n}{m}\right)^2\left(\mu_n - \bar{y}_n\right)^2 + \left(\frac{\sigma}{\sqrt{m}}-\frac{\sigma}{\sqrt{n+z}}\right)^2.\label{eqn:w2postsampling}
\end{align}
Inspecting (\ref{eqn:w2postsampling}) reveals that the second term captures the nominal EPSS: setting $m=n+z$ makes the standard deviation of the baseline posterior $\pi^f_m$ match that of the conjugate posterior  $\pi^c_n$, so the second term of (\ref{eqn:w2postsampling}) equals to zero. However, the first term of (\ref{eqn:w2postsampling}) reveals that there is an intuitive reason for the value of $M_n$ to  often be larger than the nominal EPSS: disagreements between the prior and the data as captured by $\left(\mu_n-\bar{y}_n\right)^2 = \left(z/(z+n)\right)^2\left(\mu_0-\bar{y}_n\right)^2$ mean that the two posteriors will not be centered in the same location, and the $(n/m)^2$ term in  (\ref{eqn:w2postsampling}) suggests that greater agreement  is expected to be obtained by adding further samples to the baseline prior, i.e., increasing $m$. 
Thus, our definition of $M_n$ correctly identifies that simply reporting the classical information content of the prior as determined  by its standard deviation is not sufficient: we must also take into account the impact of the prior location {\it relative to the data}. Of course, the results in Proposition \ref{THM:SAM_DIST} must also take account of $\widetilde{W}_2(m)$, see Appendix \ref{app:proof_prop} for details.  

Bayesian methodology stipulates that the extra samples must be drawn from the posterior predictive distribution, as above, but results (b) and (c)  of Proposition \ref{THM:SAM_DIST} provide   further intuition for the correctness of this approach (or rather the incorrectness of competing approaches). Result (b) supposes that $E[\bar{s}_r]=\bar{y}_n$ which is the case when sampling the future observations from the empirical distribution (bootstrap) or from the baseline posterior predictive distribution, i.e., the posterior predictive distribution under $\pi^f$ and conditioning on only the observed data $\boldsymbol{y}$. The first part of result (b) where $M_n=z$ for $\bar{y}_n \approx \mu_n$  may be considered somewhat reasonable, and is similar to what is seen in Figure \ref{fig:prior0_cases}. However,  in the current scenario, the second part of result (b) where $M_n<0$  for large $|\bar{y}_n - \mu_n|$ does not make sense both because $z>0$ and because intuitively   the prior impact is large when $\bar{y}_n$ is far from $\mu_n$. 
\citet{reimherr2014} avoided this problem by defining the EPSS so that  negative values convey a disagreement between the prior and the data, but there are limitations of their approach as discussed in Section \ref{sec:review}. Furthermore, there is always {\it some} disagreement between the data and the prior so we find it conceptually more appealing to always have positive prior impact (unless our prior is less informative than the baseline). 

Result (c)  of Proposition \ref{THM:SAM_DIST} corresponds to the case where the additional samples are drawn from the 
conjugate prior distribution $\pi^c$: just as some might argue that the future data sampling method should not be ``contaminated" by the prior, others may argue that it should not be ``contaminated" by the data! Under the scenario of the proposition, this sampling scheme yields $M_n=z$, which is at least never negative.  However, 
simply recovering the nominal EPSS regardless of whether $\mu_0=\bar{y}_n$ or $\mu_0=100\bar{y}_n$ does not convey differences in the impact of the prior, which is the purpose of having a measure of prior impact. 

In conclusion,  drawing the extra samples from the posterior predictive distribution (\ref{eqn:ppd1})-(\ref{eqn:ppd2}) seems to give the most intuitive result (i.e., (a)  of Proposition \ref{THM:SAM_DIST}), and we now briefly return to that case to gain further insights. Suppose that the observed mean of the extra samples is $\bar{s}_r= \mu_n+\alpha$, i.e., no longer exactly the theoretical mean  $\mu_n$ as in result (a) of Proposition \ref{THM:SAM_DIST}. Further investigation along the same lines reveals that for small $|\alpha|$, the result $M_n \geq z$ still holds, but for large $|\alpha|$ we obtain $M_n <0$. The proof is similar to that for Proposition \ref{THM:SAM_DIST} and is omitted. In the current example, it is undesirable for $M_n$ to be negative (as explained above), but the situation is different to that in result (b)  of Proposition \ref{THM:SAM_DIST}, because here the probability of $|\alpha|$ being large, and $M_n$ being negative, is small. This small probability represents the chance that we are unlucky and the extra samples do not well represent their true distribution, and  consequently that the conjugate prior misleadingly appears to be less informative than the baseline prior.  Lastly, we note that result (a) in Proposition \ref{THM:SAM_DIST} does not contradict the existence of situations where the MOPESS is less than the nominal EPSS $z$, such as those discussed in  Section \ref{subsec:beta-binomial} and Appendix \ref{app:small_epss}. Indeed, those cases arise due to the  variability in $\bar{s}_r$ (and additional conditions), i.e., not when assuming that $\bar{s}_r$ is  exactly equal to its theoretical mean. 

\subsection{Theoretical posterior  distribution of the OPESS}\label{sec:distance}


To study the variation in the OPESS for a given observed dataset, we now derive the theoretical distribution of the OPESS conditional on $\boldsymbol{y}$ for the preceding Gaussian conjugate posterior example. 
More generally, the distribution of the OPESS will typically be hard to derive, but it can be empirically approximated, see Algorithm~\ref{algo:OPESS_compute} Step 2. 

Lemma \ref{LEM:DIST_DIST} below gives the distribution of the distances $W_2(m)$ and $\widetilde{W}_2(m)$ conditional on $\bar{y}_n$ and $\mu$ (drawn from $\pi_n^c$ in Algorithm~\ref{algo:OPESS_compute}). The proof is given in Appendix \ref{app:proof_con_dist_dist}. We condition on both $\bar{y}_n$ and $\mu$ because then the two distances are independent which facilitates derivation of the OPESS distribution. The distance distributions conditional on only  $\bar{y}_n$ are given in Appendix \ref{app:dist_dist}. Lemma \ref{LEM:DIST_DIST} states that both distances follow shifted non-central $\chi^2$ distributions, whose non-centrality parameters depend on $\mu$ through $\lambda_m$ and $\delta_m$ (given in the lemma statement). Intuitively, the forms of $\lambda_m$ and $\delta_m$ show that the impact of the specific value of $\mu$ results from two sources: (i) its distance from $\bar{y}_n$, and (ii) its distance from $\mu_0$. Furthermore, the impact is seen to be different for $W_2(m)$ and $\widetilde{W}_2(m)$. For example, when $m=n+z$ (corresponding to the nominal sample size) then $r=z$ and $c_m^2=0$ meaning that the conditional distribution of $W_2(m)$ only depends on the discrepancy between $\mu$ and the prior mean, $|\mu-\mu_0|$, whereas this is not true for $\widetilde{W}_2(m)$.

\begin{lemma}{Conditional Distribution of Distances.}\label{LEM:DIST_DIST}
Using the same notations as in Lemma~\ref{lemma:gaussianposterior}, assume that $x_i \stackrel{\rm i.i.d.}{\sim}\mathcal{N}(\mu,\sigma^2)$, for $i=n+1,\ldots, m$. 
Then we have
\begin{align*}
    \left[W_2\left(\pi_m^f,\pi_n^c\right)\bigg|\bar{y}_n,\mu\right]& \sim \tau_m \chi_1^2 \left(\frac{\lambda_m}{\tau_m}\right) + c_m^2,
\end{align*}
where
\begin{align*}
c_m^2 &= \left(\frac{\sigma}{\sqrt{n+z}} - \frac{\sigma}{\sqrt{m}}\right)^2,\quad \tau_m = \frac{r}{m^2}\sigma^2,\\ \lambda_m&=\left(\left(\frac{z}{n+z}-\frac{r}{m}\right)(\bar{y}_n-\mu) + (1-w_n)(\mu-\mu_0)\right)^2;
\end{align*}
and
\begin{equation*}
    \left[W_2\left( \pi_m^c,\pi_n^f\right)\bigg|\bar{y}_n,\mu\right] \sim \kappa_m \chi_1^2\left(\frac{\delta_m}{\kappa_m}\right)+ \tilde{c}_m^2,
\end{equation*}
where
\begin{align*}
\tilde{c}_m^2 &= \left(\frac{\sigma}{\sqrt{m+z}} - \frac{\sigma}{\sqrt{n}}\right)^2,\quad \kappa_m = w_m^2\tau_m,\\ \delta_m&=\left(\frac{r+z}{m+z} (\bar{y}_n - \mu) + (1-w_m)(\mu-\mu_0)\right)^2.
\end{align*}
Furthermore, conditional on $\bar{y}_n$ and $\mu$, $W_2\left(\pi_m^f,\pi_n^c\right)$ and $W_2\left( \pi_m^c,\pi_n^f\right)$ are independent.
\label{lemma:con_distr_W2_gaussian}
\end{lemma}

\begin{figure}[tbph]
\includegraphics[width=0.49\textwidth,trim=0mm 10mm 0mm 10mm,clip]{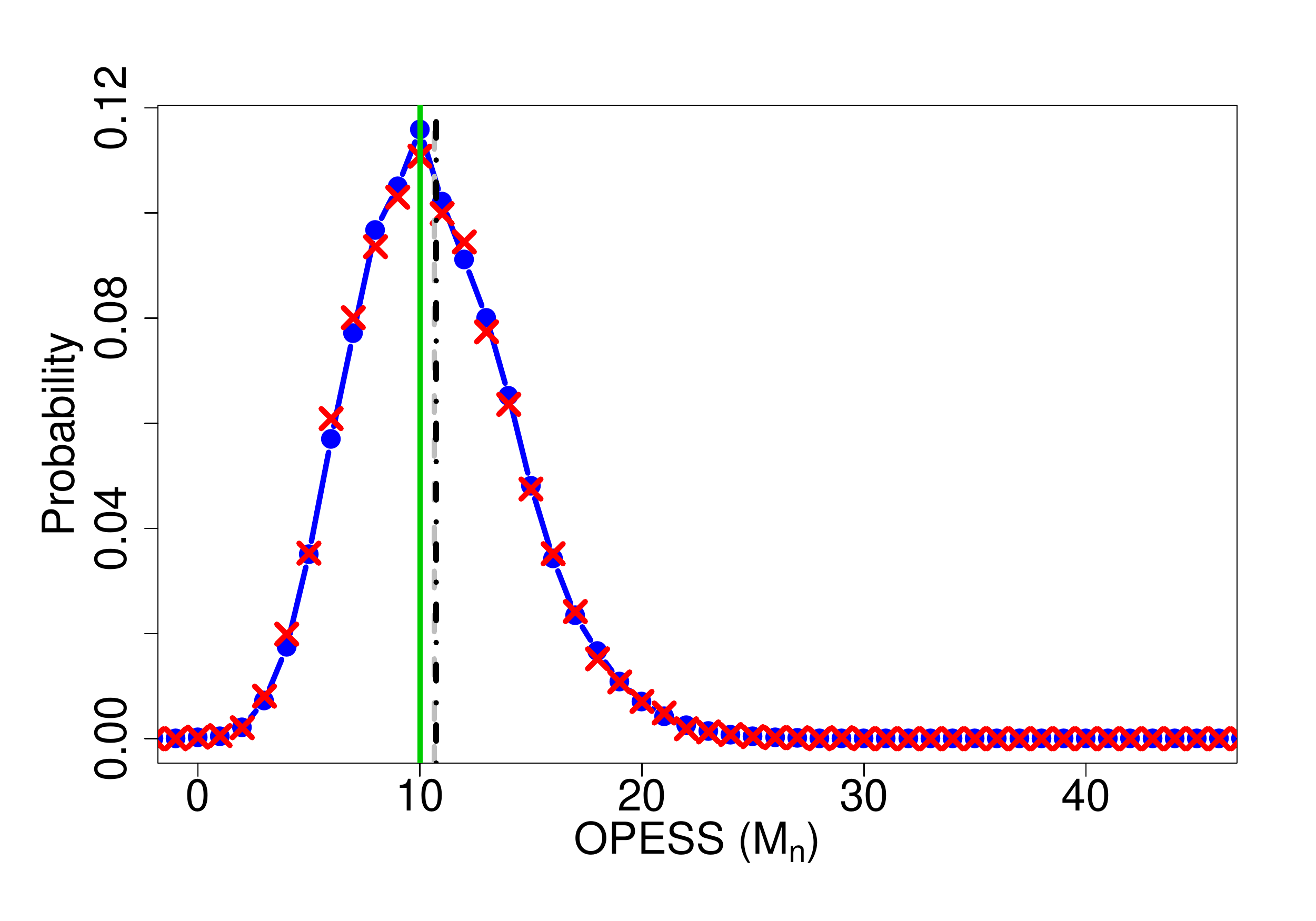}
\includegraphics[width=0.49\textwidth,trim=0mm 10mm 0mm 10mm,clip]{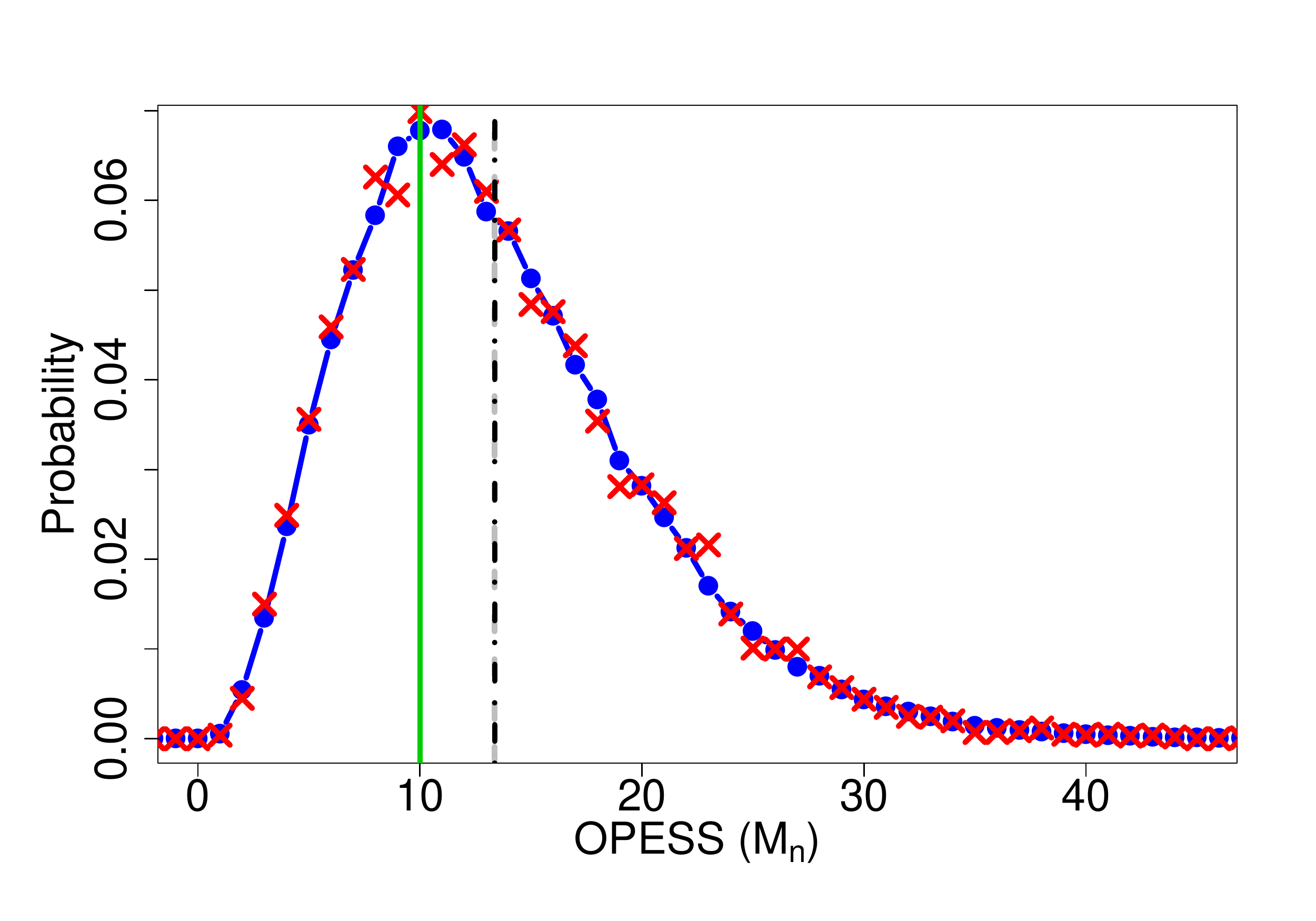}
\caption{Conditional OPESS distribution given $\bar{y}_n=\mu=0$ (left panel) and $\bar{y}_n=\mu=0.45 = 2\sigma/\sqrt{n}$ (right panel). The line-dot line shows the theoretical distribution based on  Theorem \ref{THM:EPSS_DIST} and the red crosses show the empirical distribution obtained by running Algorithm 1 Step 2 but with the specified value of $\mu$ repeatedly used for the future sample draws of Step 2(a). The solid vertical line shows the nominal EPSS (10), the dash-dot vertical line shows the mean conditional OPESS based on the theoretical distribution plotted, and the dashed vertical line shows the mean conditional  OPESS based on the empirical distribution plotted, i.e., based on Algorithm 1. In both plots the dash-dot and   dashed vertical lines closely coincide.
\label{fig:pss_dist}}
\end{figure}

Theorem \ref{THM:EPSS_DIST} below gives the posterior distribution of the OPESS conditional on $\bar{y}_n$. The proof is given in Appendix \ref{app:proof_pss_dist}. In the theorem statement, $v$ denotes a possible value of $M_n$ (e.g., in the notation $P(M_n=n|\bar{y}_n)$), and $t$ is a dummy variable for the distance corresponding to $M_n=v$, i.e., the distance $W_2(n+v)$, if $v\geq0$, and the distance $\widetilde{W}_2(n+|v|)$, otherwise. The result gives a separate expression for the case $M_n=0$ because when $m=n$ the distance between the posteriors (i.e., $W_2(n)$) is not random, meaning that an  integral over the distance dummy variable  $t$  is not required. In the case $v \in \mathbb{Z}/\{0\}$, the integrands specified are tractable because the products are truncated at   $M(t)$ and $\widetilde{M}(t)$ (defined in Appendix \ref{app:proof_pss_dist}), which  are finite for all values of $t\leq \sigma^2/(n+z)$. This truncation is possible because, for any $\epsilon>0$ and large enough $M_0$,  $W_2(m)$ and $\widetilde{W}_2(m)$ are bounded below by $\sigma^2/(n+z)-\epsilon$ and $\sigma^2/n-\epsilon$, respectively, and are less than $\sigma^2/(n+z)+\epsilon$ and $\sigma^2/n+\epsilon$, respectively, with non-negligible probability, for all $m\geq M_0$.  This means that, for any given value of $t$, either some $W_2(m)$ or $\widetilde{W}_2(m)$ is less than $t$ with probability 1 (i.e., the integrand is zero, hence the indicator functions in Theorem \ref{THM:EPSS_DIST}), or  $W_2(m)$ and $\widetilde{W}_2(m)$ are greater than $t$ with probability 1, for all $m\geq M_0$ (meaning the corresponding product terms are $1$ and can be ignored). See Appendix \ref{app:proof_pss_dist} for full details.

\begin{thm}{OPESS distribution.} \label{THM:EPSS_DIST}
Let $F_{m,\mu}$ and $\widetilde{F}_{m,\mu}$ denote the cumulative distribution function of $\chi_1^2 \left(\frac{\lambda_m}{\tau_m}\right)$ and $\chi_1^2\left(\frac{\delta_m}{\kappa_m}\right)$, respectively. Furthermore, let $h_{m,\mu}$ and $\tilde{h}_{m,\mu}$ denote the conditional probability density function of $W_2(m)=W_2\left(\pi_m^f,\pi_n^c\right)$ and $\widetilde{W}_2(m)=W_2\left( \pi_m^c,\pi_n^f\right)$, respectively, as given in Lemma \ref{lemma:con_distr_W2_gaussian}. Lastly, let $g(t,\mu,v,M,\widetilde{M})$ denote the function that gives $P(\underset{m}{\min}(W_2(m),\widetilde{W}_2(m))>t|\bar{y}_n,\mu)$ multiplied by the appropriate density for $t$, i.e.,
\begin{align*}
\left\{
  \begin{array}{cc}
\displaystyle \prod_{\substack{m=n+1\\m\neq v+n}}^{M(t)} (1 - F_{m,\mu}(t_m))\prod_{m=n+1}^{\widetilde{M}(t)} (1 - \widetilde{F}_{m,\mu}(\tilde{t}_m)) h_{v+n,\mu}(t), & \text{if } v\in \mathbb{Z}_{> 0}; \\&\\
\displaystyle \prod_{m=n+1}^{M(t)} (1 - F_{m,\mu}(t_m))\prod_{\substack{m=n+1\\m\neq |v|+n}}^{\widetilde{M}(t)} (1 - \widetilde{F}_{m,\mu}(\tilde{t}_m)) \tilde{h}_{|v|+n,\mu}(t), & \text{if } v\in \mathbb{Z}_{<0}; \\&\\
\displaystyle \prod_{m=n+1}^{M(t)} (1 - F_{m,\mu}(t_m))\prod_{m=n+1}^{\widetilde{M}(t)} (1 - \widetilde{F}_{m,\mu}(\tilde{t}_m)), & \text{if } v=0,\end{array}\right.
\end{align*}
where $t_m = (t - c_m^2)/\tau_m$ and $\tilde{t}_m = (t - \tilde{c}_m^2)/\kappa_m$, and $M$ and $\widetilde{M}$ are known functions (see Appendix \ref{app:proof_pss_dist}).
Then $P(M_n=v|\bar{y}_n)$ is given by
\begin{align*}\left\{
    \begin{array}{cc}\displaystyle\int_\mathbb{R} \int_{T_v} 1_{\{W_2(n),\widetilde{W}_2(n),\sigma^2/(n+z)\geq t\}} g(t,\mu,v,M,\widetilde{M}) dt\, \pi(\mu|\bar{y}_n)d\mu, & \text{if }v \in \mathbb{Z}/\{0\},\\&\\
    \displaystyle 1_{\{W_2(n) \leq \sigma^2/(n+z)\}}\int_\mathbb{R}g(W_2(n),\mu,0,M,\widetilde{M})\pi(\mu|\bar{y}_n)d\mu,  & \text{if } v=0, \\
    \end{array} \right.
\end{align*}
where $T_v$ is $\mathbb{R}_{\geq c_m^2}$ if $v>0$ and $\mathbb{R}_{\geq \tilde{c}_m^2}$ otherwise, and 
$M(t)$ and $\widetilde{M}(t)$ are finite integers for all values of $t\leq \sigma^2/(n+z)$.
\end{thm}

Figure \ref{fig:pss_dist} shows two examples of the {\it conditional} posterior distribution of the OPESS given $\bar{y}_n$ {\it and} $\mu$, where $\bar{y}_n=\mu=0$ in the left panel and $\bar{y}_n=\mu=0.45=2\sigma/\sqrt{n}$ in the right panel. We plot the conditional posterior distribution to gain intuition about how the particular draw of $\mu$ from $\pi^c$ impacts the conditional distribution of the OPESS. This is important because in reality the value of $\mu$ is fixed but unknown, and it is therefore valuable to understand how the distribution of the OPESS changes when we simulate the future samples based on different fixed choices of $\mu$. In Figure \ref{fig:pss_dist}, the line-dot density is a close Monte Carlo approximation to the theoretical conditional density of the OPESS given $\bar{y}_n$ and $\mu$, and was obtained by simulating from the theoretical conditional distributions of $W_2(m)$ and $\widetilde{W}_2(m)$ and averaging the resulting values of the integrand given in Theorem \ref{THM:EPSS_DIST} (except in the case $P(M_n=0)$ for which no Monte Carlo approximation is needed). The red crosses show the empirical distribution of the OPESS obtained by directly applying the first two steps of Algorithm 1, except with the modification that $\mu$ is fixed in Step 2(a). Figure \ref{fig:pss_dist} illustrates that for $\bar{y}_n$ (and $\mu$) farther from the prior mean $\mu_0=0$ (right panel) the  conditional OPESS distribution has larger mean and is  more right-skewed. This corroborates the numerical results seen in Figure \ref{fig:prior0_cases}. For some values of $\bar{y}_n$ and $\mu$ the conditional posterior distribution of the OPESS is bi-modal, with one mode at positive values and one at negative values (not shown). For other  $\bar{y}_n$ and $\mu$, there is a mode at $M_n=0$, which is relevant to the case where the MOPESS is less than the nominal EPSS, a scenario that is discussed further in Section \ref{subsec:beta-binomial} and Appendix \ref{app:small_epss}.

\section{Further numerical examples}
\label{subsec:nonnormal_numerical}
In this section we present numerical simulation studies that mimic the
conditions in Section \ref{subsec:gaussian_numerical} but for the
Beta-Binomial  and simple linear regression models.

\subsection{Beta-Binomial model}
\label{subsec:beta-binomial}
Suppose 
$\{y_i, 1 \leq i \leq n \}$
are independent observations taking values in the set $\{0, 1\}$. The unknown parameter $\theta$ is the probability that $y_i = 1$.
We set the informative prior to be $\pi^c(\theta) \equiv \text{Beta}(\alpha, \beta)$, where $\alpha, \beta$ are known hyperparameters, and the baseline prior to be
$\pi^f(\theta)\equiv \text{Beta}(1, 1)$. As in Section $\ref{sec:gaussian}$,
$x_i^{(m)} = y_i$ for $i =1, \dots, n$, 
and $x_i^{(m)} | \, \theta \stackrel{\text{i.i.d.}}{\sim} \text{Bernoulli}(\theta)$ for $i=n+1, \dots, m$ (but since $\theta$ is unknown it is drawn from its posterior distribution when computing the MOPESS, see Algorithm \ref{algo:OPESS_compute} Step 2(a)).
Let $\pi_n^A$ and $\pi_m^A$
denote the posterior distribution using the original data $\boldsymbol{y} = (y_1,\dots,y_n)^T$ and the expanded dataset $\boldsymbol{x}^{(m)}$, respectively,  under prior $A$. Also define $(F_n^A)^{-1}$ and $(F_m^A)^{-1}$ to be the
quantile functions associated with these posterior densities. Then it can be shown that the $2$-Wasserstein distance
between $\pi_n^A$ and $\pi_m^A$ is $\left(\int_0^1 ((F_n^A)^{-1}(u) - (F_m^A)^{-1}(u))^2 du\right)^{1/2}$, see Theorem 2 in \citet{cambanis1976inequalities}. Unfortunately, this distance
cannot be expressed in closed form in the case of Beta distributions, but it can be approximated to high precision using numerical integration, which is the approach we take.

In our simulations we set $\alpha = \beta = 5$, which corresponds to a nominal EPSS of $\alpha + \beta - 2 = 8$.
The subtraction of $2$ highlights that the standard nominal EPSS is relative to the prior sample size of the flat prior Beta$(1,1)$, which is also our baseline prior $\pi^f$. We sample
$1,000$ datasets of size $n=20$ with replacement from the sex ratio dataset presented in Section 2.4 of \cite{gelman2013bayesian}. The dataset consists of the biological sexes of 980 babies born to mothers with placenta previa: 437 of the babies are female, a proportion of  0.446
of the total.

The top left panel of Figure \ref{fig:example-beta-binom} shows the MOPESS values obtained across the $1,000$ simulations. The $\overline{M}_n$ estimates have a similar  pattern as in the Gaussian conjugate model of Section \ref{subsec:gaussian_numerical},
except that $\overline{M}_n$ is less than the nominal EPSS for datasets with $\bar{y}_n = 0.5$. The top right panel of Figure \ref{fig:example-beta-binom} shows that the median of the
posterior distribution of $M_n$ is similar to the mean. It also shows that the posterior distribution is much wider for datasets with means less than the
prior mean of $0.5$. 
 This is due to the fact that the posterior distribution of $\theta$ is right skewed when $\bar{y}_n$ is less than
the prior mean, and this skewness is consequently reflected in the future observations $\boldsymbol{x}_{\text{\tiny L}}^{\text{all}}$, which are simulated conditional on a draw of $\theta$. Hypothetical future datasets 
with means that are far greater than $\bar{y}_n$ will result in positive OPESS, while future datasets with means near or smaller than $\bar{y}_n$
will result in negative OPESS. 
That the posterior skewness is reflected in the resulting MOPESS distribution is a strength of our approach, and in particular of integrating over the posterior uncertainty of $\theta$ when calculating our measure
of prior impact. 
Indeed, large spread in the posterior distribution of $M_n$ simply reflects genuine uncertainty about the number of extra samples that need to be collected and combined with the baseline posterior in order to match the posterior  under $\pi^c$. 


\begin{figure}[t]
\includegraphics[width=0.49\textwidth,trim=0mm 10mm 0mm 0mm,clip]{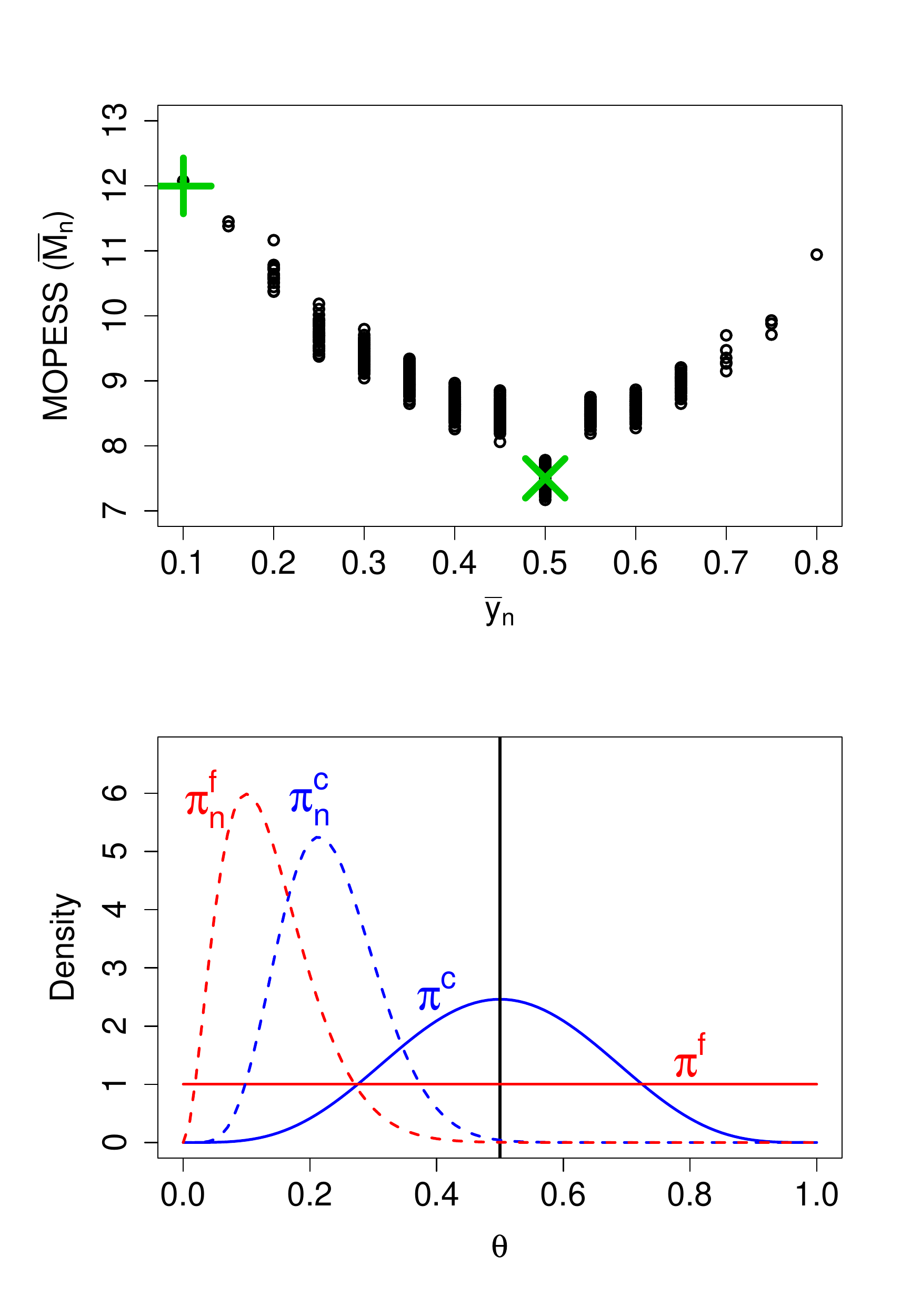}
\includegraphics[width=0.49\textwidth,trim=0mm 10mm 0mm 0mm,clip]{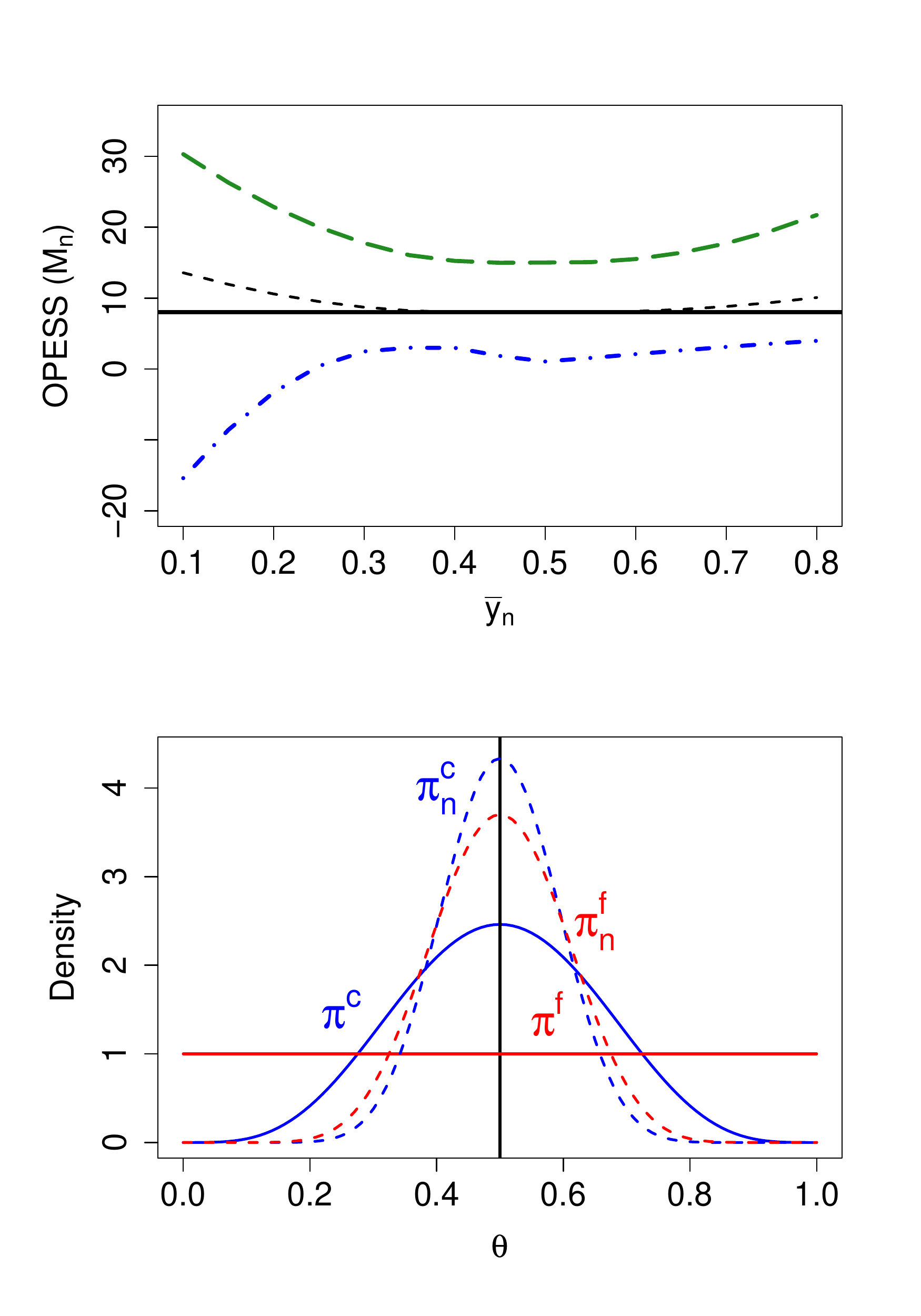}
\caption{(Top left) MOPESS ($\overline{M}_n$) as a function of observed data mean.  Each point corresponds to one of the $1,000$ simulated datasets. The 
green ``+'' and green cross symbols indicate datasets which had large and small discrepancies with the prior mean, respectively. (Top right) Quantiles of the posterior distribution of $M_n$ as a function of observed data mean including the median (dashed curve in black), 95\% quantile (long-dashed curve in green), and 5\% quantile (dash-dot curve in blue). The horizontal solid line shows the nominal EPSS of $8$. 
(Bottom left) Comparison of posteriors $\pi_n^c$ and $\pi_n^f$ with respective priors when the observed data mean is 0.1, indicated in the top left graph 
with a green ``+''. (Bottom right) Comparison of posteriors $\pi_n^c$ and $\pi_n^f$ with respective priors when the observed data mean is 0.5, indicated in the top left graph as a 
green cross.}
\label{fig:example-beta-binom}
\end{figure}

Next, we examine the relationship between the MOPESS and the nominal EPSS in two cases, namely, those where the prior mean and $\bar{y}_n$ are highly discrepant and  perfectly aligned, respectively. 
The top left panel of Figure \ref{fig:example-beta-binom} indicates a simulated dataset for which $\bar{y}_n = 0.1$ (green ``$+$''), and the bottom left panel shows the corresponding initial posterior distributions $\pi_n^c$ and $\pi_n^f$ (as well as the prior distributions). 
In this case, the MOPESS value is high (around $12$) because the initial posteriors are very different. The bottom right panel of Figure \ref{fig:example-beta-binom} show analogous plots for a dataset with $\bar{y}_n = 0.5$, as labeled with green cross on the top left panel. In this case,  the initial posteriors are very similar, which is why the MOPESS value is low (approximately $7.5$, the variation about which is Monte Carlo error). In particular, the MOPESS value is {\it less} than the nominal EPSS of $8$, a phenomenon that did not occur in the Gaussian conjugate model example of Section  \ref{subsec:gaussian_numerical} for any of the simulated datasets. The low MOPESS value occurs here due to specific circumstances that, in this case, arise due to the discreteness of the data and the future data, as we now explain. The data mean $\bar{y}_n$ exactly matches the mean of the prior $\pi^c$, which in turn makes the means of $\pi_n^c$ and $\pi_n^f$ exactly equal. Thus, $\pi_n^c$ and $\pi_n^f$ are very similar to begin with, and it is unclear whether adding more samples to one of these posteriors will further reduce the distance between them. Adding more samples to the baseline posterior $\pi_n^f$  could reduce the width of the distribution, and therefore may lead to greater agreement with $\pi_n^c$. However, the discrete nature of the data means that one additional sample with value one or zero will necessarily move the posterior mean away from $0.5$, therefore potentially increasing the 2-Wasserstein distance between the two posteriors. Of course, if we draw an even number of extra samples then their average may be close to $0.5$, so a reduction in the width of the baseline posterior $\pi_n^f$  may be achieved without any substantial change in the mean. However, based on the nominal EPSS value, the approximate number of extra samples needed for matching the posterior widths is $8$, but the probability of achieving an average of $0.5$ (or very close to this) when drawing around $8$ samples is not sufficiently high, and consequently the distance $W(n)$ is often smaller than $W_2(m)$ and $\widetilde{W}_2(m)$ for all $m > n$. Thus, for many simulations of $\boldsymbol{x}_{\text{\tiny L}}^{\text{all}}$, we have $M_n=0$, meaning that the MOPESS $\overline{M}_n$ is shrunk towards zero. 

In summary, we may expect the MOPESS to be less than the nominal EPSS when the means of the initial posteriors $\pi_n^c$ and $\pi_n^f$ are very similar relative to the size of $z$.
Furthermore, adding extra samples to one posterior  may not reduce the 2-Wasserstein distance between the two posteriors because: (i) if few extra samples are added then the variability in their mean can introduce discrepancies between the posterior means, and (ii) adding many extra samples will introduce discrepancies in the spreads since the initial discrepancy will be over-corrected. Thus, often the smallest distance between the posteriors is achieved when $M_n=0$, and consequently $\overline{M}_n$ is small. In Appendix \ref{app:small_epss} we demonstrate that this phenomenon {\it can} occur in the Gaussian conjugate model example if $n \gg z$ (whereas in Section \ref{subsec:gaussian_numerical} we set $n=2z$).  \citet{reimherr2014} discussed a related phenomenon. 

\subsection{Simple linear regression model}
\label{subsec:linearregression}

We now consider the setting of a simple linear regression model: 
\begin{align} \label{eqn:reg-model}
\begin{split}
   Y_i | \boldsymbol{\beta}, X_i = x_i & \sim \mathcal{N}(\beta_1 + \beta_2 x_i, \sigma^2),\quad X_i \sim \mathcal{N}(0, 1), 
   \end{split}
\end{align}
for $i=1,\ldots, n$, where $\sigma^2$ is known, and $\boldsymbol{\beta} = (\beta_1, \beta_2)'$ are the unknown model parameters. We note that in simple linear regression models, distributional assumptions on covariates are typically not made. We assume that $X_i \sim \mathcal{N}(0, 1)$ is known for simplicity of statements of the algorithm for generating hypothetical samples. Let our informative prior $\pi^c(\boldsymbol{\beta})$ be: 
\begin{align*}
   \pi^c(\boldsymbol{\beta}) & = \mathcal{N}\lp\boldsymbol{\eta}_0,\Sigma_0 \rp, 
\quad \text{where} \quad  \Sigma_0  = 
   \begin{bmatrix}
   \tau_1^2 & 0 \\
   0 & \tau_2^2
   \end{bmatrix},
\end{align*}
and $\boldsymbol{\eta}_0 = (\mu_0, \gamma_0)'$ and $\tau_1, \tau_2$ are known hyperparameters. Thus, the nominal EPSS for $\beta_i$ is given by $\sigma^2 / \tau_i^2 = z_i$, for $i \in [1,2]$. 
We set the baseline prior to be $\pi^f(\boldsymbol{\beta}) \propto 1$. Define the $m^{\rm th}$ set of hypothetical samples as $\{(y_i^{(m)}, x_i^{(m)}), i \in \{1,\dots,m\}\}$
with $\{(y_i^{(m)}, x_i^{(m)}) = (y_i, x_i), i \in \{1,\dots,n\}\}$ for all $m$. 
For $i > n$, the hypothetical samples $(y_i^{(m)}, x_i^{(m)})$ are generated from \eqref{eqn:reg-model} conditional on a draw of $\boldsymbol{\beta}$ from the posterior
distribution $\pi_n^c$. Given that the models that we consider here are all Gaussian conjugates,  the posteriors  $\{\pi_u^c(\boldsymbol{\beta}), \pi_u^f(\boldsymbol{\beta}), n\leq u\leq L\}$ are also Gaussian. Closed expressions for the posterior distributions and corresponding Wasserstein distances are given in Appendix~\ref{appendix:expressions_linear_regression}. Thus, it is straightforward to apply Algorithm~\ref{algo:OPESS_compute} to compute the MOPESS.

Our linear regression model simulation study is similar in design to that for the Beta-Binomial model in Section \ref{subsec:beta-binomial}. We observe $n=20$ samples from the model \eqref{eqn:reg-model},
with $\sigma^2 = 1$, and $\beta_1 = \beta_2 = 0$. We set $z_1 = z_2 = 10$, so the nominal EPSS of $\pi^c$ is 10.
\begin{figure}[t]
\includegraphics[width=0.49\textwidth,trim=0mm 0mm 0mm 0mm,clip]{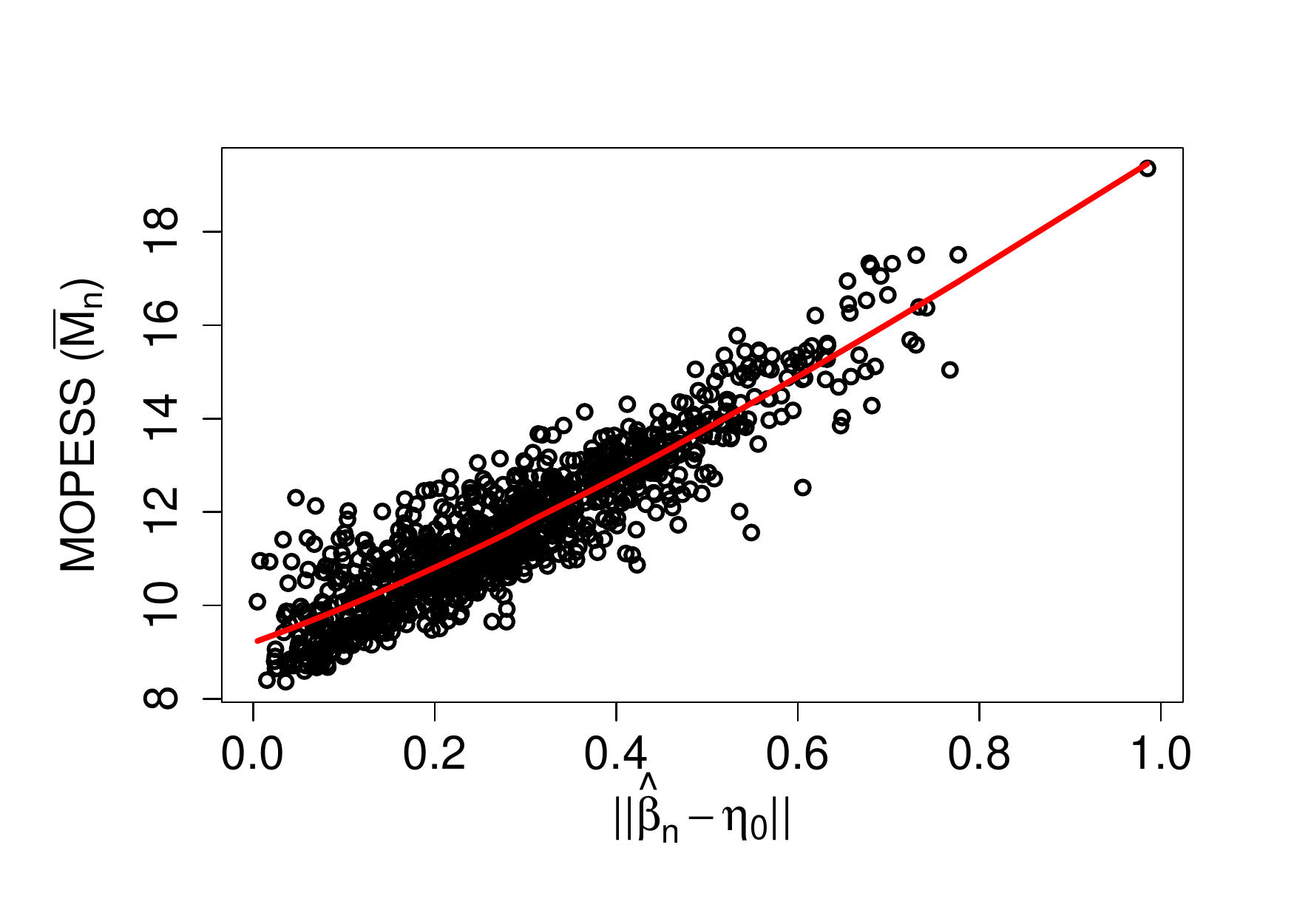}
\includegraphics[width=0.49\textwidth,trim=0mm 0mm 0mm 0mm,clip]{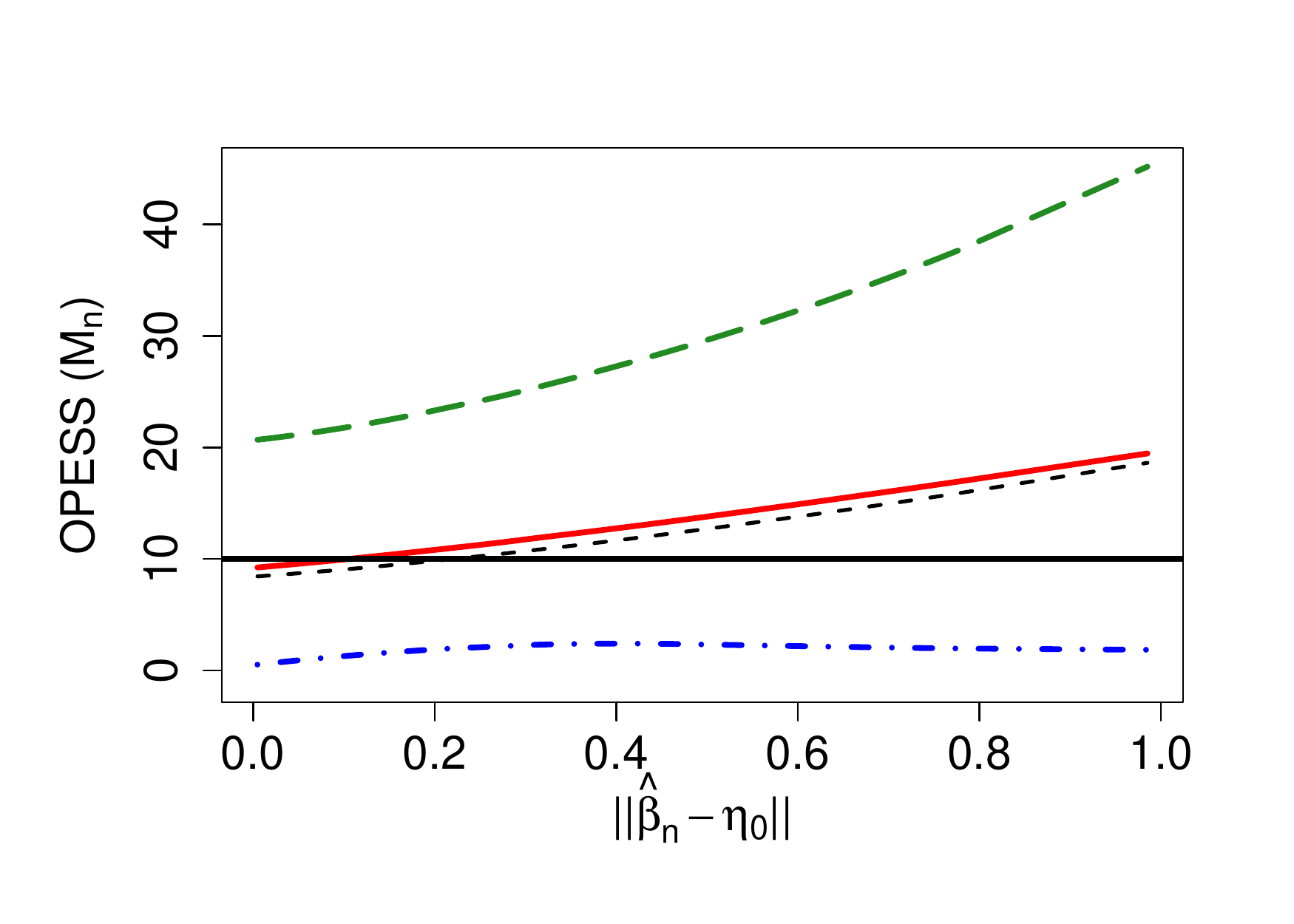}
\caption{(Left) MOPESS ($\overline{M}_n$) as a function of the $L_2$-norm of $\hat{\boldsymbol{\beta}}_n - \boldsymbol{\eta}_0$ (in this case $\boldsymbol{\eta}_0 = 0$).  Each point corresponds to one of the $1,000$ simulated datasets. (Right) Quantiles of the posterior distribution of $M_n$ as a function of the $L_2$-norm including the median (dashed curve), 95\% quantile (long-dashed curve), and 5\% quantile (dash-dot curve). The solid red line is the same as the red line in the left panel and the horizontal line shows the nominal EPSS of $10$.  }
\label{fig:regression-results}
\end{figure}
\begin{figure}[t]
\includegraphics[width=0.49\textwidth,trim=0mm 0mm 0mm 0mm,clip]{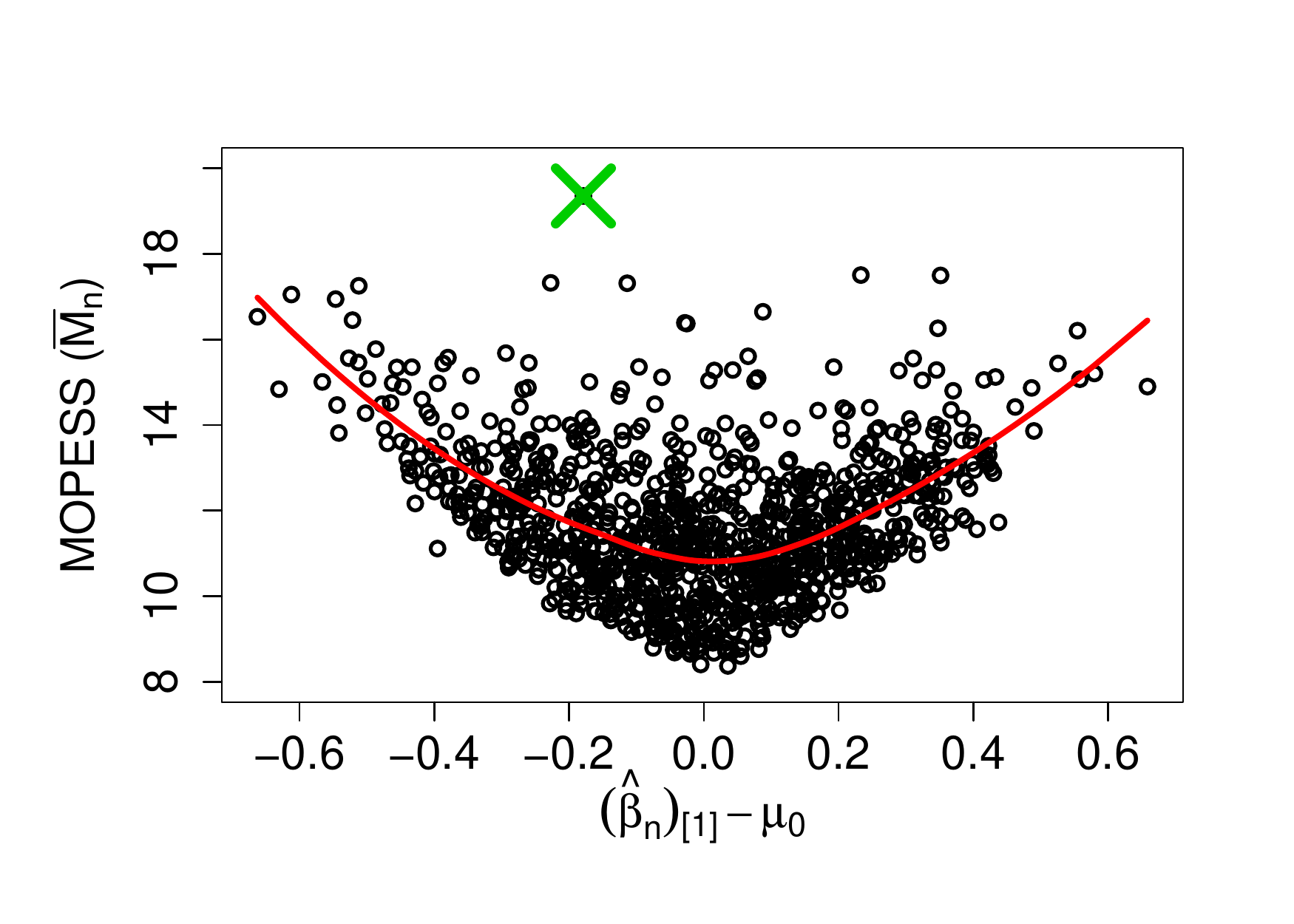}
\includegraphics[width=0.49\textwidth,trim=0mm 0mm 0mm 0mm,clip]{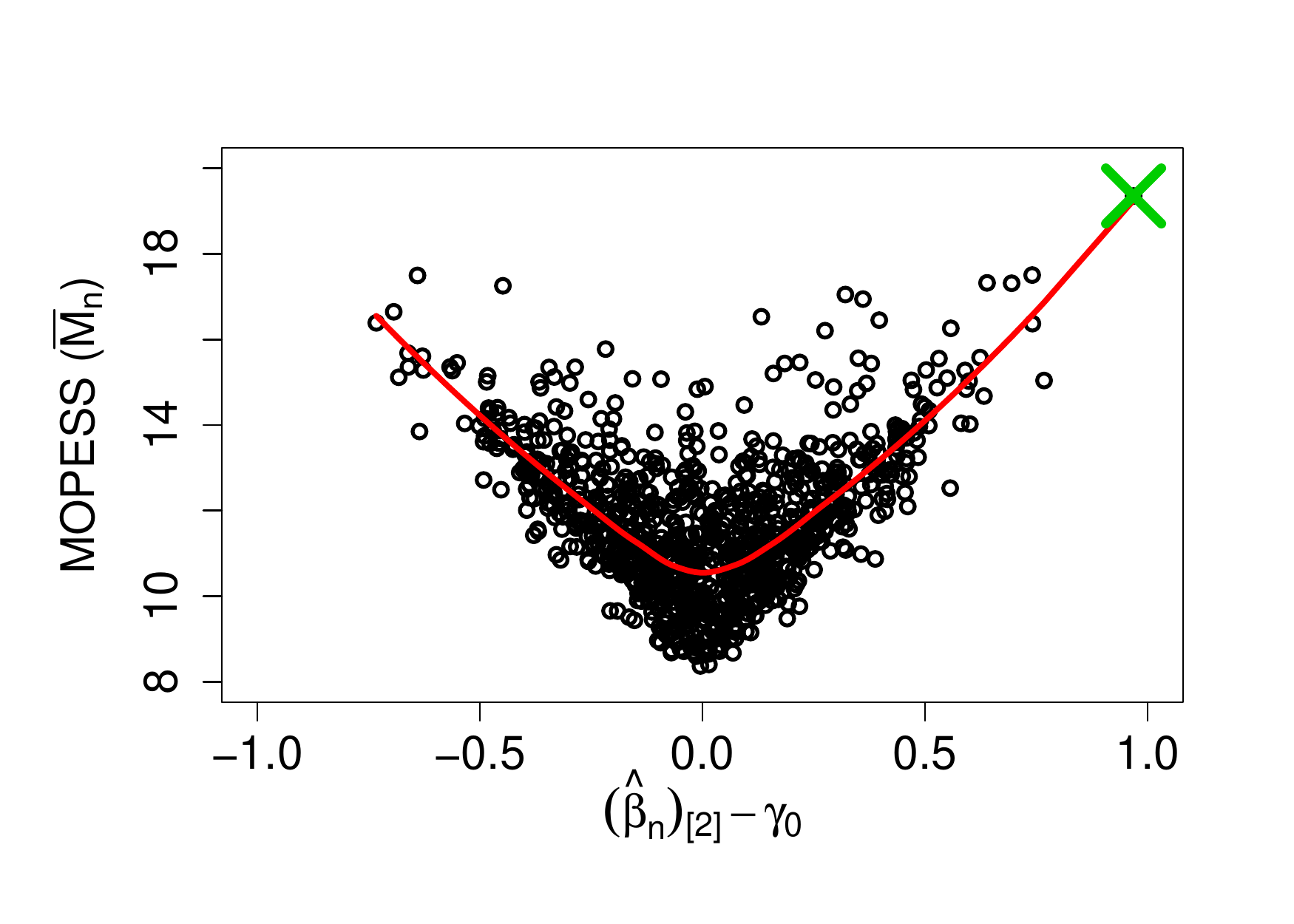}
\caption{(Left) MOPESS ($\overline{M}_n$) as a function of the of $(\hat{\boldsymbol{\beta}}_n)_{[1]} - \mu_0$ (in this case $\mu_0 = 0$). (Right) MOPESS ($\overline{M}_n$) as a function of the of $(\hat{\boldsymbol{\beta}}_n)_{[2]} - \gamma_0$ (in this case $\gamma_0 = 0$).  Each point in both graphs corresponds to one of the $1,000$ simulated datasets. The cross 
symbol indicates the maximum observed MOPESS across all $1,000$ simulation studies, $\approx 19$.}
\label{fig:marginal-regression}
\end{figure}
As can be seen in Figure \ref{fig:regression-results}, the MOPESS increases with $\norm{\hat{\boldsymbol{\beta}}_n - \boldsymbol{\eta}_0}_2$, i.e., the $L_2$-norm of the ordinary least squares estimator less the prior mean. 
We use a one-dimesional summary of the two-dimensional measure $\hat{\boldsymbol{\beta}}_n - \boldsymbol{\eta}_0$ to ease visualization, and
to account for the fact that the joint prior can be influential even if only one element of $\hat{\boldsymbol{\beta}}_n$ disagrees with the corresponding marginal prior. 

Indeed, Figure \ref{fig:marginal-regression} shows how the MOPESS can vary even when conditioning on a small interval of $(\hat{\boldsymbol{\beta}}_n)_{[1]} - \mu_0$, or
$(\hat{\boldsymbol{\beta}}_n)_{[2]} - \gamma_0$. In the left panel, we see that for values of $(\hat{\boldsymbol{\beta}}_n)_{[1]} - \mu_0$ that are
near zero, there is still quite a range of MOPESS values. Thus the MOPESS is influenced not only by $(\hat{\boldsymbol{\beta}}_n)_{[1]} - \mu_0$, but also the
disagreement between $\gamma_0$ and $(\hat{\boldsymbol{\beta}}_n)_{[2]}$. For example, the maximum MOPESS value in the left panel, indicated by the 
cross, occurs at a value of $(\hat{\boldsymbol{\beta}}_n)_{[1]} - \mu_0$ that is far from extreme. But turning to the right panel in Figure \ref{fig:marginal-regression},
which shows the MOPESS versus $(\hat{\boldsymbol{\beta}}_n)_{[2]} - \gamma_0$, we see that the maximum MOPESS occurs at the maximum observed value of 
$(\hat{\boldsymbol{\beta}}_n)_{[2]} - \gamma_0$. In summary, the MOPESS generalizes to two dimensions as we would expect it to, with joint dependence on 
$\hat{\boldsymbol{\beta}}_n - \eta_0$.

\section{Discussion and future work}\label{sec:discussion}

\subsection{Minimum distance as a measure of information relevance}

We highlight a potentially key phenomenon: for different realizations of $\boldsymbol{x}_{\text{\tiny L}}^{\text{all}}$, the minimum distance achieved is different. In other words, in some realizations of the OPESS $M_n$ and the future samples $\boldsymbol{x}_{\text{\tiny L}}^{\text{all}}$, the posteriors $\pi_n=q_{\pi}(\cdot|\bf{y})$ and $\pi_{M_n+n}^b=q_{\pi^b}(\cdot|\boldsymbol{x}^{(M_n+n)})$ (or  $\pi_{|M_n|+n}$ and $\pi_n^b$) are more similar than in other realizations. In practice, we have found that the smallest minimum distance between the posteriors is typically achieved for realizations in which $M_n$ equals the nominal EPSS. This is illustrated by Figure \ref{fig:pss_against_dist} which shows values of $M_n$ plotted against minimum distance for the Gaussian conjugate example discussed in Section \ref{sec:illustration}. Realizations that lead to large minimum distance tend to more strongly favor small and large values of $M_n$ over intermediate values.

The value of the minimum distance represents how closely it is possible to achieve the goal of matching inference under the baseline prior $\pi^b$ to our actual inference. Thus, when we report $\overline{M}_n$ we should also report the average minimum distance (or some other summary of the minimum distance distribution) because this summarizes the {\it quality} of the posterior matches and therefore represents the relevance of the reported MOPESS. In future work we will further investigate this notion of information {\it relevance} or {\it reliability}. 

\begin{figure}[t]
\centering
\includegraphics[width=0.6\textwidth,trim=0mm 10mm 0mm 17mm,clip]{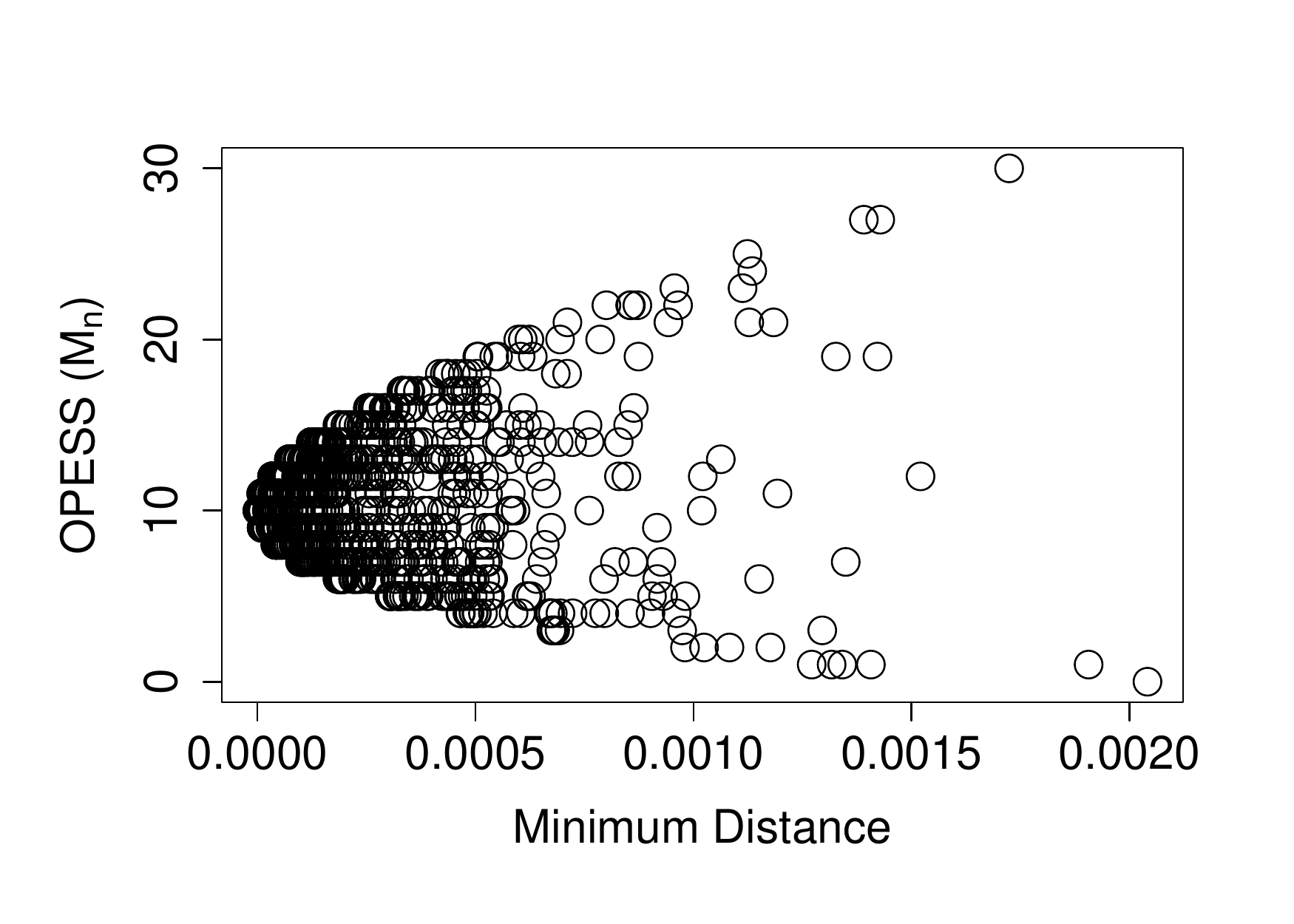}
\caption{OPESS ($M_n$) against $\text{min}(W,\widetilde{W})$, for a fixed value of $\bar{y}_n=0.057$ in the Gaussian conjugate example of Section \ref{sec:illustration}. The nominal EPSS is $10$ and $\overline{M}_n=10.47$. \label{fig:pss_against_dist}}
\end{figure}

\subsection{Low-impact priors}

Our approach allows negative MOPESS values, like that of \citet{reimherr2014}. For instance, if we set $\pi^b=\mathcal{N}(\mu_0,0.5\lambda_0^2)$ in the Gaussian example of  Section \ref{subsec:gaussian_numerical}  then $\overline{M}_n$ would usually be negative. Thus, in more complex situations our framework can be used to assess if a prior is a good choice for a low-impact prior, i.e., if it is less informative than a standard baseline prior.

\subsection{Generalizations to Non-conjugate Models}

Although the examples that we study in this paper are all conjugate models, our framework (Algorithm~\ref{algo:OPESS_compute}) is generalizable to non-conjugate models. In contrast to conjugate models for which the posterior distribution is derived analytically, for most non-conjugate models, the posterior distribution needs to be represented by posterior samples based on Bayesian computational techniques such as Markov chain Monte Carlo (MCMC) algorithms. However, the number of expensive MCMC runs needed is small because the posterior distributions considered in Algorithm~\ref{algo:OPESS_compute} (with different $m$ values) are similar, meaning that importance sampling can provide fast and reliable samples once we have samples from $q_{\pi}(\cdot|\boldsymbol{y})$ and $q_{\pi^b}(\cdot|\boldsymbol{y})$. If $q_{\pi^b}(\cdot|\boldsymbol{y})$ is also reasonably similar to $q_{\pi}(\cdot|\boldsymbol{y})$ then MCMC is only needed for sampling from the latter, and this computation would typically have already been preformed in the original analysis.

\subsection{Future Work and Applications}

Besides the generalizations to non-conjugate models as described above, we plan to extend our current framework and develop efficient computational algorithms for more complex models such as hierarchical models and models with nuisance parameters. Hierarchical models have already been discussed in the introduction and in related literature (see Section~\ref{sec:review}). The latter extension can be explained using the regression example given in Section~\ref{subsec:linearregression}. In a regression model, typically the slope is the parameter of interest and the intercept is the nuisance parameter. 
In Section~\ref{subsec:linearregression} we used our measure to quantify prior impact on the {\it joint} inference for the parameters. 
However, it is also of interest to  quantify the (joint) prior impact on inference for the slope parameter alone, i.e., ignoring any impact the prior has on inference for the intercept parameter that does not impact inference for the slope parameter. In future work we will investigate this problem, especially in cases where the parameter of interest and the nuisance parameter are correlated  either a priori or due to the data. This might offer important insights for problems where Bayesian inference, which is supposed to give inferences that form a compromise between the prior and the likelihood, gives counter-intuitive results~\citep{xie2013incorporating,chen2020geometric}. Furthermore, we plan to apply our methods to evaluate prior impacts in the astronomical meta-analysis  problem~\citep{chen2019calibration} discussed in  Section~\ref{subsec:linearregression}, in which nuisance parameters play a crucial role, and better understanding the impact they have on final inferences could prove valuable in many scientific analyses and for the implementation of multi-instrument astronomical observation.

\bibliographystyle{plainnat}
\bibliography{priorsamplesize}

\section*{Acknowledgement}

This work is supported by NSF DMS-1811083 (PI: Yang Chen, 2018 - 2021). The authors thank Prof. Xiao-Li Meng from Harvard University for helpful discussions and Dr. Vinay Kashyap from the Harvard-Smithsonian Center for Astrophysics (CfA) for collaborating on the astronomical instrument calibration problem.

\appendix

\section{Proof of Proposition \ref{THM:SAM_DIST}}\label{app:proof_prop}

In  part (a), $\gamma=\mu_n$ and from (\ref{eq:illustrationw2a}) we have
\begin{align}
    \Wm = \left(\frac{n}{m}\right)^2\left(\mu_n - \bar{y}_n\right)^2 + \left(\frac{\sigma}{\sqrt{n+z}}-\frac{\sigma}{\sqrt{m}}\right)^2.\label{eqn:app_w2postsampling}
\end{align}
Let $m^*$ denote the minimizer of $W_2(m)$ and suppose that $m^* < n+z$. For $m=n+z$ the second term of (\ref{eqn:app_w2postsampling}) is zero and the first term is smaller than for $m^*$, thus $W_2(n+z) < W_2(m^*)$, a contradiction. Therefore, $m^*\geq n+z$. (Aside: note that $m^*=n+z$ holds if $\mu_n=\bar{y}_n$, i.e., if $\mu_0=\bar{y}_n$.) 
Furthermore, referring to (\ref{eq:illustrationw2b}), for $m > n$, we have
\begin{align}
\Wmflip
=  \left(\bar{y}_n-\mu_n\right)^2 +  \left(\frac{\sigma}{\sqrt{n}} - \frac{\sigma}{\sqrt{m+z}} \right)^2 > \Wm.\label{eqn:w2flippostsampling}
\end{align}
To see this note the following. If  $m=n$ then the second term on the right-hand side of (\ref{eqn:app_w2postsampling}) is equal to the corresponding standard deviation term in (\ref{eqn:w2flippostsampling}). For $n<m\leq n+z$, the standard deviation term in (\ref{eqn:app_w2postsampling}) decreases whereas that in (\ref{eqn:w2flippostsampling}) increases.  Regarding the case $m>n+z$, $\sigma/\sqrt{n}$ is larger than $\sigma/\sqrt{n+z}$ and the value being subtracted from these terms  is smaller in (\ref{eqn:w2flippostsampling}) than in (\ref{eqn:app_w2postsampling}) ($\sigma/\sqrt{m+z}$ versus $\sigma/\sqrt{m}$), and thus again the standard deviation term is larger in (\ref{eqn:w2flippostsampling}) . Thus, for all $m > n$, the standard deviation term in (\ref{eqn:w2flippostsampling}) is larger than that in  (\ref{eqn:app_w2postsampling}). This verifies the inequality in (\ref{eqn:w2flippostsampling})
because the first term (\ref{eqn:w2flippostsampling}) is necessarily  larger than that of (\ref{eqn:app_w2postsampling}) for $m>n$. Thus, $M_n\geq z$.

  In part (b), $\gamma=\bar{y}_n$ and from (\ref{eq:illustrationw2a}) and (\ref{eq:illustrationw2b}) we have
\begin{align}
\Wm
&= \left(\mu_n -\bar{y}_n\right)^2 +  \left(\frac{\sigma}{\sqrt{m}} - \frac{\sigma}{\sqrt{n+z}} \right)^2,\label{eq:illustration_ybar1}\\
\Wmflip
&= \left(\frac{n+z}{m+z}\right)^2\left(\bar{y}_n-\mu_n\right)^2 +  \left(\frac{\sigma}{\sqrt{n}} - \frac{\sigma}{\sqrt{m+z}} \right)^2.\label{eq:illustration_ybar2}
\end{align}
Clearly (\ref{eq:illustration_ybar1}) is minimized at $m=n+z$, because in that case the second term on the right hand side is zero (and the first does not depend on $m$). The second term on the right-hand side of (\ref{eq:illustration_ybar2}) is monotonically increasing for $m\geq n$. Thus, setting $\epsilon_s$ to be any value such that $\epsilon_s^2 < \left((\sigma/\sqrt{n}) - (\sigma/\sqrt{n+z}) \right)^2$ yields the first part of result (b). The first term on the right-hand side of (\ref{eq:illustration_ybar2}) converges to zero as $m$ increases, and the second term is bounded above by $\sigma^2/n$. Thus, choosing $\epsilon_l^2>\sigma^2/n$, the second part of result (b) follows.   

In part (c), $\gamma=\mu_0$ and we have
\begin{align}
\Wm
&= \left(\mu_n -\mu_{n,m}\right)^2 +  \left(\frac{\sigma}{\sqrt{m}} - \frac{\sigma}{\sqrt{n+z}} \right)^2,\\
\Wmflip
&= \left(\bar{y}-\mu_{n,m+z}\right)^2 +  \left(\frac{\sigma}{\sqrt{n}} - \frac{\sigma}{\sqrt{m+z}} \right)^2,
\end{align}
where $\mu_{n,m}=w_{n,m}\bar{y}_n+(1-w_{n,m})\mu_0$, and $w_{n,m}=n/m$. Since $m=n+z$ gives $\Wm=0$, and $\Wmflip >0$ for all $m\geq n$, we have $M_n=z$.

\section{Proof of Lemma \ref{LEM:DIST_DIST}}\label{app:proof_con_dist_dist}

We denote the additional samples collected by $s_1,\dots,s_r$, i.e., $\{x_1,\dots,x_m\}= \{y_1,\dots,$\\ $y_n,s_1,\dots,s_r\}$, and  write $\bar{s}_r$ to denote $\frac{1}{r}\sum_{i=1}^r s_i$. Thus, we have 
\begin{align}
\bar{s}_r|\bar{y}_n,\mu\sim N\left(\mu,\frac{\sigma^2}{r}\right).\label{eqn:sr_app}
\end{align} 
From (\ref{eq:illustrationw2a}) we have
\begin{align}
\Wm
&= \left(\mu_n -\frac{n}{m}\overline{y}_n-\frac{r}{m}\overline{s}_r\right)^2 +  c_m^2\\
&= \left(w_n\bar{y}_n+(1-w_n)\mu_0 -\frac{n}{m}\overline{y}_n-\frac{r}{m}\overline{s}_r\right)^2 +  c_m^2\\
&= \left(\left(\frac{n}{n+z}-1+1-\frac{n}{m}\right)\bar{y}_n+(1-w_n)\mu_0 -\frac{r}{m}\overline{s}_r\right)^2 +  c_m^2\\
&=\left(\left(\frac{r}{m}-\frac{z}{n+z}\right)\bar{y}_n+(1-w_n)\mu_0 -\frac{r}{m}\overline{s}_r\right)^2 +  c_m^2\label{eqn:app_w2}
\end{align}
The conditional distribution of $W_2(m)$ given $\mu$ and $\bar{y}_n$ stated in Lemma \ref{LEM:DIST_DIST} then follows from the distribution of $\bar{s}_r$ (the only random quantity in (\ref{eqn:app_w2})). The proof for the conditional distribution of $\Wmflip$ is similar and is omitted. The proof of Theorem \ref{thm:dist_dist} in Appendix \ref{app:dist_dist}, which gives the distance distributions conditional on only $\bar{y}_n$, relies on analogous arguments  and is also omitted.

\section{Proof of Theorem \ref{THM:EPSS_DIST}}\label{app:proof_pss_dist}

Firstly, suppose that $v > 0$. Let $W={\rm argmin}_{m \geq n}\Wm$ and $\widetilde{W}={\rm argmin}_{m \geq n}\widetilde{W}_2(m)$. Then $P(M_n=v|\bar{y}_n) = P(W\geq W_2(v+n), \widetilde{W}\geq W_2(v+n)|\bar{y}_n)$ can be expressed as 
\begin{align}
  &\int_\mathbb{R} \int_{T_v} 1_{\{\sigma^2/(n+z)\geq t\}}\prod_{\substack{m=n\\m\neq v+n}}^{\infty} (1 - F_{m,\mu}(t_m))\prod_{m=n}^{\infty} (1 - \widetilde{F}_{m,\mu}(\tilde{t}_m)) h_{v+n,\mu}(t)dt  \pi(\mu|\bar{y})d\mu \label{eqn:infsum}\\
    =& \int_\mathbb{R} \int_{T_v}1_{\{\sigma^2/(n+z)\geq t\}} \prod_{\substack{m=n\\m\neq v+n}}^{M(t)} (1 - F_{m,\mu}(t_m))\prod_{m=n}^{\widetilde{M}(t)} (1 - \widetilde{F}_{m,\mu}(\tilde{t}_m)) h_{v+n,\mu}(t)dt  \pi(\mu|\bar{y})d\mu,\label{eqn:finsum}
    \end{align}
where, for $\sigma^2/(n+z)\geq t$,   $M(t)$ denotes the minimum value of $m \in \mathbb{Z}_{\geq n+z}$ such that
  \begin{align}\label{eqn:min_m_t}
\left(\frac{\sigma}{\sqrt{m}} - \frac{\sigma}{\sqrt{n+z}}\right)^2 > t
   \end{align}
   and  $\widetilde{M}(t)$ denotes the minimum value of $m \in \mathbb{Z}_{\geq n+1}$ such that
     \begin{align}\label{eqn:min_tilde_m_t}
\left(\frac{\sigma}{\sqrt{m+z}} - \frac{\sigma}{\sqrt{n}}\right)^2 > t.
   \end{align}
In particular,  for all $ t \leq \sigma^2/(n+z)$, $M(t)$ and $\widetilde{M}(t)$ are both finite, and for  $m>M(t)$ we have $P(W_2(m)>t|\bar{y},\mu)=0$, and similarly for $m>\widetilde{M}(t)$ we have $P(\widetilde{W}_2(m)>t|\bar{y},\mu)=0$.  This demonstrates the equality of (\ref{eqn:infsum}) and (\ref{eqn:finsum}). 
The indicator $1_{\{\sigma^2/(n+z)\geq t\}}$ is needed because  if  $\sigma^2/(n+z)< t$, then the above arguments do not hold and there exists $m^* > v+n$  such that $W_2(m^*) < t$ or $\widetilde{W}_2(m^*) < t$ with probability $1$, meaning that $P(W\geq t, \widetilde{W}\geq t|\bar{y}_n,\mu)=0$. For completeness we set $M(t)=\widetilde{M}(t)=\infty$ if  $\sigma^2/(n+z)< t$. Next,  $P(W_2(n)>t|\bar{y},\mu), P(\widetilde{W}_2(n)>t|\bar{y},\mu)\in\{0,1\}$, meaning that (\ref{eqn:finsum}) can be written as
    \begin{align}
\int_\mathbb{R} \int_{T_v} 1_{\{W_2(n),\widetilde{W}_2(n),\sigma^2/(n+z)\geq t\}} g(t,\mu,v,M(t),\widetilde{M}(t)) dt \pi(\mu|\bar{y})d\mu\label{eqn:pmf}
\end{align}
as in Theorem \ref{THM:EPSS_DIST}. The proof in the case $v < 0$ is analogous. 

Lastly, if $v=0$, then  essentially the same derivation holds except that the integral over $t$ is no longer required and $P(M_n=0|\bar{y}_n)$ simplifies to 
\begin{align}
1_{\{W_2(n) \leq \sigma^2/(n+z)\}}\int_\mathbb{R} \prod_{m=n+1}^{M(W_2(n))} (1 - F_{m,\mu}(s_m))\prod_{m=n+1}^{\widetilde{M}(W_2(n))} (1 - \widetilde{F}_{m,\mu}(\tilde{s}_m))\pi(\mu|\bar{y})d\mu.
\end{align}

\section{Conditional Distribution of Distances}\label{app:dist_dist}

\begin{thm}{Distribution of Distances.}\label{thm:dist_dist}
Using the same notations as in Lemma~\ref{lemma:gaussianposterior}, assume that $x_i \stackrel{\rm i.i.d.}{\sim}\mathcal{N}(\mu,\sigma^2)$, $i=n+1,\ldots, m$, where $\mu$ is a random sample from the posterior distribution $\pi_n^c$; then we have
\begin{align*}
    \left[W_2\left(\pi_m^f,\pi_n^c\right)\bigg|\boldsymbol{y}\right]& \sim \tau_m \chi_1^2 \left(\frac{\lambda}{m^2\tau_m}\right) + \left(\frac{\sigma}{\sqrt{n+z}} - \frac{\sigma}{\sqrt{m}}\right)^2\\
    &= \tau_m \left( \frac{\sqrt{\lambda}}{m\sqrt{\tau_m}} + \mathcal{Z} \right)^2 + \left(\frac{\sigma}{\sqrt{n+z}} - \frac{\sigma}{\sqrt{m}}\right)^2,
\end{align*}
where
\begin{align*}
\tau_m& = \frac{r^2}{m^2}\left(\frac{w_n}{n}+\frac{1}{r}\right)\sigma^2, \quad \lambda=\left(n(1-w_n)(\bar{y}_n-\mu_0)\right)^2;
\end{align*}
and
\begin{equation*}
    \left[W_2\left( \pi_m^c,\pi_n^f\right)\bigg|\boldsymbol{y}\right] \sim \kappa_m \chi_1^2\left(\frac{\delta}{\kappa_m}\right)+ \left(\frac{\sigma}{\sqrt{m+z}} - \frac{\sigma}{\sqrt{n}}\right)^2,
\end{equation*}
where
\begin{align*}
\kappa_m& = w_m^2\tau_m, \quad \delta=\frac{\lambda}{n^2}.
\end{align*}
\label{theorem:distr_W2_gaussian}
\end{thm}

\section{Small MOPESS scenario in the conjugate Gaussian model example}\label{app:small_epss}


Consider the Gaussian conjugate model of Section \ref{subsec:gaussian_numerical}. Figure \ref{fig:small-epss}
shows a scenario in which we the  MOPESS is less than the nominal EPSS for a number of simulated datasets. The left-hand side of Figure \ref{fig:small-epss} shows
$\overline{M}_n$ versus $\bar{y}_n$ for a small numerical simulation study with nearly the same settings as 
in section \ref{subsec:gaussian_numerical}, but with $z=4$ instead of $z=10$. We can see that for $\left|\bar{y}_n\right| \leq 0.2$ the MOPESS is estimated to be less than
4, modulo Monte Carlo approximation error. 
A histogram of the OPESS is shown in the right-hand plot of \ref{fig:small-epss} for a single simulated dataset with $\bar{y}_n = 1.8\times10^{-5}$ that appears as a single point
on the left-hand graph. A notable feature of the right-hand plot is that there are no negative OPESS realizations,
and there is a preponderance of OPESS realizations of zero. These zeros pull the MOPESS below the nominal EPSS. Excluding posterior draws of $M_n = 0$
leads to a MOPESS of near 4 for the smallest $\left|\bar{y}_n\right|$. Thus we investigate the probability of the event that $M_n = 0$,
when there is a small discrepancy between $\mu_n$ and $\bar{y}_n$.

Suppose that $\bar{y}_n = \mu_n + \varepsilon$ where $\varepsilon$ is small. The event that $[M_n = 0]$ is the event that 
$[W_2(n+r) \geq W_2(n)]$ and $[\tilde{W}_2(n+r) \geq W_2(n)]$ for all $r > 0$. If $P(W_2(n+r) > W_2(n)\, |\, \bar{y}_n) > 0$ or
$P(W_2(n+r) \geq W_2(n)\,|\,\bar{y}_n) > 0$ for all $r$, the probability of $M_n = 0$ would be zero, but as we can see from 
Appendix \ref{app:proof_pss_dist} Equations \eqref{eqn:min_m_t} and \eqref{eqn:min_tilde_m_t} show that as $r$ (equivalently, $m$)
increases, there exists an $\hat{r}$ such that $P(W_2(n+r) \geq W_2(n)\,|\,\bar{y}_n) = 0$ for all $r$ greater than
$\hat{r}$ and an $r^\prime$ such that $P(W_2(n+r) \geq W_2(n)\,|\,\bar{y}_n) = 0$ for all $r > r^\prime$.

Let us first consider the probability $p(r) = P(W_2(n+r) < W_2(n)\,|\,\bar{y}_n)$, 
where $r=m-n$, as before. We have:
\begin{align}
 \label{eqn:small_pss1}
\begin{split}
 p(r) = P\left( \left((\mu_n - \bar{s}_r) \left(\tfrac{r}{n + r}\right) - \varepsilon \tfrac{n}{n+r}\right)^2 < \varepsilon^2 + c_n^2 - c_{n+r}^2 \;\middle|\; \bar{y}_n\right)
\end{split}
\end{align}
where $\bar{s}_r$ is the mean value of the additional hypothetical samples, defined in Section \ref{sec:sampling_dist}, and 
$c_{n+r}^2 = \left(\tfrac{1}{\sqrt{n + z}} - \tfrac{1}{\sqrt{n + r}}\right)^2$, defined in Section \ref{sec:distance}. 

The right hand side of the inequality in \eqref{eqn:small_pss1}:
\begin{align*}
   \varepsilon^2 + \left(\tfrac{1}{\sqrt{n + z}} - \tfrac{1}{\sqrt{n}}\right)^2 - \left(\tfrac{1}{\sqrt{n + z}} - \tfrac{1}{\sqrt{n + r}}\right)^2
\end{align*}
must be greater than zero for $p(r) > 0$. The greatest value for $r$ such that the inequality is greater than zero is $\hat{r}$.
Given that $P(M_n = 0\,|\,\bar{y}_n) = \prod_{r=1}^\infty (1 - p(r))$, fewer terms less than 1 in the product will increase the probability that
$M_n = 0$. This gives some intuition as to why $n \gg z$ leads to MOPESS less than nominal EPSS when $\varepsilon$ is small:
$\varepsilon^2 + c_n^2 - c_{n+r}^2 \approx 0$, leading to $p(r) \approx 0$.

The expression \eqref{eqn:small_pss1} also shows that when $z$ is moderately sized compared to $n$, we can also estimate
MOPESS to be smaller than EPSS. The key aspect of \eqref{eqn:small_pss1} is that the right hand side of the inequality depends quadratically on $\varepsilon$, while the left-hand side of the inequality depends on both a quadratic and linear function of $\varepsilon$. The consequence is that the event $[W_2(n+r) < W_2(n)]$ only occurs if $r$ is not too large {\it and} $\bar{s}_r$ happens to be sufficiently close to $\mu_n$.

As for $P(\widetilde{W}_2(n+r) < W_2(n)\,|\,\bar{y}_n)$, or $\tilde{p}(r)$, we can see immediately that $\tilde{p}(r) = 0$ for 
all $r$ for small $\varepsilon$:
\begin{align}
\tilde{p}(r) = P\left(\left(\tfrac{r}{m+z}(\mu_n - \bar{s}_r) + \varepsilon\right)^2 < c_n^2 - \tilde{c}_{n+r}^2 + \varepsilon^2  \;\middle|\; \bar{y}_n \right)
\end{align}
The expression $c_n^2 - \tilde{c}_{n+r}^2 + \varepsilon^2$ is 
$\left(\tfrac{1}{\sqrt{n}} - \tfrac{1}{\sqrt{n + r}}\right)^2 - \left(\tfrac{1}{\sqrt{n}} - \tfrac{1}{\sqrt{n + r + z}}\right)^2 + \varepsilon^2$. 
So in order for $\tilde{p}(r) > 0$, $\tilde{c}_{n+r}^2 - c_n^2 < \varepsilon^2$.
In our simulation study $n=20, z=4$, $\varepsilon^2 > 1\times10^{-4}$ in order for $\tilde{p}(1) > 0$. For the 
example shown in the histogram in the right-hand panel of Figure \ref{fig:small-epss} $\varepsilon^2 \approx 1\times10^{-10}$,
so $\tilde{p}(r) = 0 \, \forall \, r > 1$ which agrees with the paucity of negative $M_n$ realizations shown in the empirical histogram.

\begin{figure}[t]
\includegraphics[width=0.49\textwidth,trim=0mm 0mm 0mm 0mm,clip]{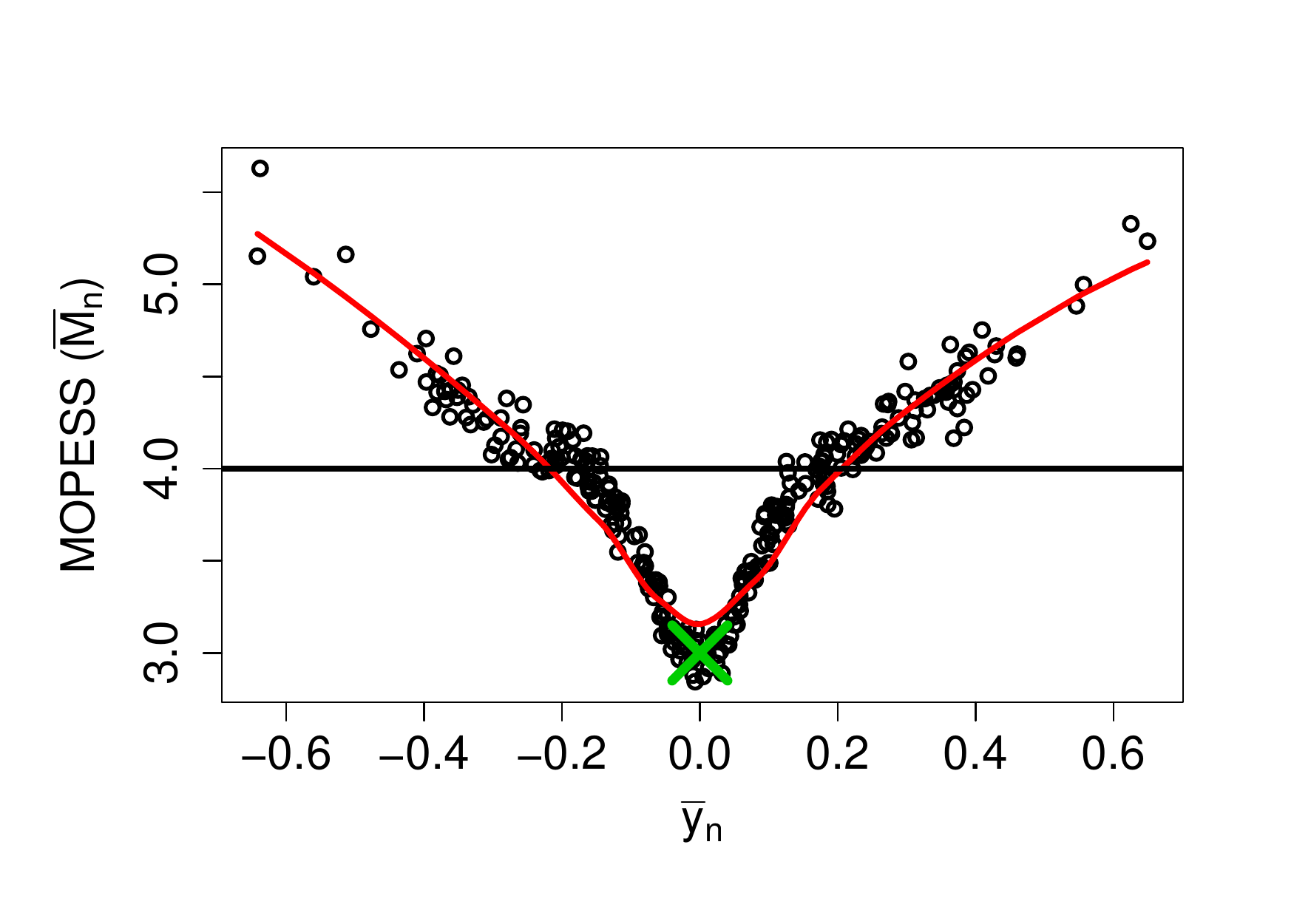}
\includegraphics[width=0.49\textwidth,trim=0mm 0mm 0mm 0mm,clip]{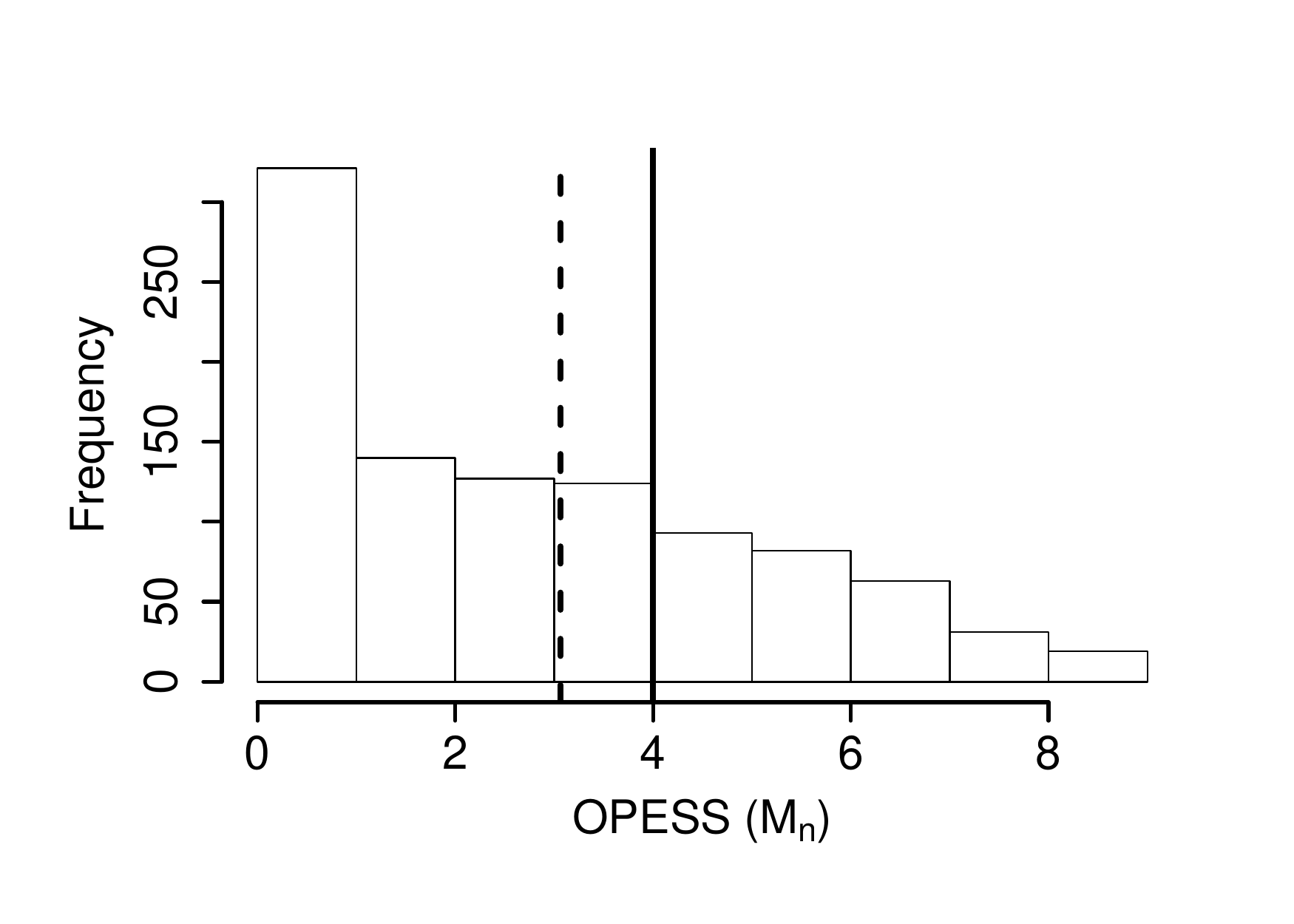}
\caption{(Left) MOPESS ($\overline{M}_n$) as a function of $\bar{y}_n$ for $n=20$ and $z=4$ for 300 bootstrap samples for the Gaussian conjugate example.  Each point corresponds to one of the 300 simulated datasets. The solid horizontal line at $y = 4$ indicates the nominal EPSS, while the red line is the loess smoothed relationship between $\overline{M}_n$ and $\bar{y}_n$. The green cross indicates a single dataset with the minimum observed $\bar{y}_n$ across all 300 simulations, whose posterior over $M_n$ is shown in the right-hand graph (Right) 
$1,000$ posterior realizations of $M_n$ for a the given simulation with $\bar{y}_n = 1.8\times10^{-5}$. The black solid line indicates the nominal EPSS of $z=4$, while the 
dashed black line indicates MOPESS of $\approx 3$.}
\label{fig:small-epss}
\end{figure}

\section{Choice of Discrepancy Measure}\label{sec:app_distance}

In this Section, we derive the relationship between the Kullback-Leibler (KL) divergence, an alternative discrepancy measure, and the Wasserstein distance for the Gaussian conjugate model. We thereby demonstrate the generality of our proposed framework in Section~\ref{sec:general_form} and point out the connections between the Wasserstein distance and the KL divergence. 

Let $\mu$ and $\nu$ be probability measures on a metric space $\mathcal{M}$. The KL divergence from $\nu$ to $\mu$, defined as ${\rm KL}(\mu, \nu) = \int \mu \log\frac{\mu}{\nu}$, quantifies the information gain if $\nu$ substitutes $\mu$. If $\nu$ is a prior distribution and $\mu$ is a posterior distribution, ${\rm KL}(\mu, \nu)$ describes the change in belief due to the data (likelihood). Consider the Gaussian conjugate model in Lemma~\ref{lemma:gaussianposterior}. The KL divergence from $\prior_m^f$ to $\prior_n^c$ is given by
\begin{equation}
{\rm KL}\left(\prior_n^c, \prior_m^f\right) = \frac{1}{2} \left\{ \frac{m}{n+z} + \frac{m}{\sigma^2} D_{m, n} - 1 + \log \frac{n+z}{m}\right\},\label{eq:KLgaussian_app}
\end{equation}
where $D_{m,n}=\left[w_n\overline{y}_n + (1-w_n) \mu_0 - \overline{x}_m\right]^2$. The $p^{\rm th}$ Wasserstein distance between $\pi_m^f$ and $\pi_n^c$ with $p=2$ is
\begin{equation}
W_2\left( \pi_m^f,\pi_n^c \right)
= D_{m, n} +  \left(\frac{\sigma}{\sqrt{m}} - \frac{\sigma}{\sqrt{n+z}} \right)^2.
\end{equation}

\subsection{Generality of Proposed Framework}
We use this Gaussian example to illustrate the generality of the framework we propose in Section~\ref{sec:general_form} for defining the prior influence by changing the discrepancy measure and the way of generating hypothetical samples. Theorem~\ref{theorem:empiricalmeasure} gives the minimizer of the expected KL divergence when the hypothetical samples are obtained by sampling the observations with replacement independently in the Gaussian conjugate model. The expectation is taken with respect to the sampling of hypothetical observations conditioning on the actual observations. The result reveals that with the discrepancy measure being the KL divergence and the generation of hypothetical samples being the simple bootstrap, minimizing the expected KL divergence can also give valid characterizations of the prior impact, which takes into account of the prior-likelihood (mis)alignment. We note that in this example, the minimization is computed after taking the expectation, which is different from the previous sections. This is only for ease of derivations of analytical forms. And the interpretation based on the analytical solution, which is not the same as the ideal solution under our rigorous definition, is intuitive.
\begin{thm}
\label{theorem:empiricalmeasure}
Assume that $x_1,\ldots, x_m$ are independently sampled with replacement from $\{y_1,\ldots, y_n\}$. Then $\mathbb{E} [{\rm KL}\left(\prior_n^c, \prior_m^f\right)]$ is minimized at 
\begin{align*}
m^* &= \left\lfloor (n+z)\left[1 + \frac{z^2(\bar{y}_n-\mu_0)^2}{(n+z)\sigma^2}\right]^{-1} \right\rceil,
\end{align*}
where $\lfloor\rceil$ denotes the nearest integer of a real number. Note that $m^*$ is less than $n$ if $(\bar{y}_n-\mu_0)^2 > \sigma^2 \left(n^{-1}+z^{-1}\right)$ and greater than $n$ otherwise. This $m^*$ is also the minimizer of $\mathbb{E} [{\rm KL}\left(\prior_n^c, \prior_m^f\right)]$ when $x_i\stackrel{\rm i.i. d.}{\sim} p_{\prior}(\cdot|\boldsymbol{y}_{1:n})$, where $\pi$ is reference prior; and when $\pi$ is conjugate prior, $m^*=n+z$. 
\end{thm}
From this theorem, we know that when the prior and the likelihood roughly aligns, the prior contributes a ``positive'' sample size, and vice versa.

\begin{proof}
For the Gaussian conjugate model, from the definition of $D_{m, n}$, the expected value of $D_{m, n}$ (thus the KL divergence) only depends on the first two moments of $\overline{x}_m$ through
\begin{equation*}
\mathbb{E} (D_{m, n}) =  \left[w_n\overline{y}_n + (1-w_n) \mu_0 - \mathbb{E}(\overline{x}_m)\right]^2 + \mathbb{V}(\overline{x}_m).
\end{equation*}
Now we derive the minimizer of the expected KL divergence under different sampling schemes of $\boldsymbol{x}_{1:m}$.
\begin{itemize}
\item Sampling with replacement. $x_i \stackrel{\rm i.i.d.}{\sim} \hat{F}_n(\cdot)$. \label{create_hypo_samples_bootstrap}
\begin{align*}
&\mathbb{E}(\bar{x}_m) = \bar{y}_n, \mathbb{V}(\bar{x}_m) = \frac{1}{m} S_n^2;\mathbb{E}[D_{m, n}] = (1-w_n)^2 (\bar{y}_n-\mu_0)^2+\frac{S_n^2}{m}.\\
&\frac{\partial \mathbb{E}\left[{\rm KL}\left(\prior_n^c, \prior_m^f\right)\right]}{\partial m} = \frac{1}{2}\left\{\frac{1}{n+z} + (1-w_n)^2\frac{(\bar{y}_n-\mu_0)^2}{\sigma^2} -\frac{1}{m}\right\}. 
\end{align*}
Therefore, $\mathbb{E}\left[{\rm KL}\left(\prior_n^c, \prior_m^f\right)\right]$ is minimized at $m = \left[\frac{1}{n+z}+(1-w_n)^2 \frac{(\bar{y}_n-\mu_0)^2}{\sigma^2}\right]^{-1}$.
\item Posterior predictive (i). $x_i\stackrel{\rm i.i. d.}{\sim} p_{\prior}(\cdot|\boldsymbol{y}_{1:n})$, where $\pi$ is conjugate prior.\label{create_hypo_sample_pp1}
\begin{align*}
&\mathbb{E}(\overline{x}_m) =  \mathbb{E}(\mu(\theta) | \boldsymbol{y}_{1:n}) =w_n\overline{y}_n + (1-w_n) \mu_0;\\
&\mathbb{V}(\overline{x}_m) =\frac{1}{m} \left[\mathbb{E}(\sigma^2(\theta)|\boldsymbol{y}_{1:n}) + \mathbb{V}(\mu(\theta)|\boldsymbol{y}_{1:n})\right]=\frac{\sigma^2}{m}+ \frac{1}{m} \frac{\sigma^2}{n+z}.\\
&\frac{\partial \mathbb{E}\left[{\rm KL}\left(\prior_n^c, \prior_m^f\right)\right]}{\partial m} = \frac{1}{2}\left\{\frac{1}{n+z}  -\frac{1}{m}\right\}. 
\end{align*}
Therefore, $\mathbb{E}\left[{\rm KL}\left(\prior_n^c, \prior_m^f\right)\right]$ is minimized at $m = n+z$.
\item Posterior predictive (ii). $x_i\stackrel{\rm i.i. d.}{\sim} p_{\prior}(\cdot|\boldsymbol{y}_{1:n})$, where $\pi$ is reference prior.\label{create_hypo_sample_pp2}
\begin{align*}
&\mathbb{E}(\overline{x}_m) =  \mathbb{E}(\mu(\theta) | \boldsymbol{y}_{1:n}) =\overline{y}_n;\\
&\mathbb{V}(\overline{x}_m) =\frac{1}{m} \left[\mathbb{E}(\sigma^2(\theta)|\boldsymbol{y}_{1:n}) + \mathbb{V}(\mu(\theta)|\boldsymbol{y}_{1:n})\right]=\frac{\sigma^2}{m}+ \frac{1}{m} \frac{\sigma^2}{n}.\\
&\frac{\partial \mathbb{E}\left[{\rm KL}\left(\prior_n^c, \prior_m^f\right)\right]}{\partial m} = \frac{1}{2}\left\{\frac{1}{n+z} +(1-w_n)^2\frac{(\bar{y}_n-\mu_0)^2}{\sigma^2} -\frac{1}{m}\right\}.
\end{align*}
Therefore, $\mathbb{E}\left[{\rm KL}\left(\prior_n^c, \prior_m^f\right)\right]$ is minimized at $m = \left[\frac{1}{n+z}+(1-w_n)^2 \frac{(\bar{y}_n-\mu_0)^2}{\sigma^2}\right]^{-1}$.
\end{itemize}
\end{proof}

\subsection{Comparison of Wasserstein and KL Distances}
Next, we show with the conjugate Gaussian example the relationship between the two discrepancy measures for the two posterior distributions: the KL divergence and the Wasserstein distance, in Proposition~\ref{proposition:KLandFisherinfor_Gaussian}.

\begin{proposition} From Equations~\eqref{eq:KLgaussian_app} and~\eqref{eq:W2gaussian1} in the Gaussian conjugate model, we have
\begin{equation}
\frac{{\rm KL}\left(\prior_n^c, \prior_m^f\right)}{W_2\left( \pi_m^f,\pi_n^c \right) } = \frac{1}{2}\ \frac{m}{\sigma^2}\left\{1 + \frac{\sigma^2}{m}\frac{2\frac{\sqrt{m}}{\sqrt{n+z}} - 2 + \log \frac{n+z}{m}}{W_2\left( \pi_m^f,\pi_n^c \right) }\right\} \geq \frac{I_m}{2},
\label{eqn:KL_l2_gaussian}
\end{equation}
where $I_m = \frac{m}{\sigma^2}$ is the Fisher information. (i) The equality in~\eqref{eqn:KL_l2_gaussian} holds if and only if $m=n+z$. (ii) For any fixed/finite $n$, as $m\rightarrow\infty$, the equality in~\eqref{eqn:KL_l2_gaussian} holds asymptotically. 
\label{proposition:KLandFisherinfor_Gaussian}
\end{proposition}
\begin{proof}
The proof of inequality~\eqref{eqn:KL_l2_gaussian} is due to the fact that $2\frac{\sqrt{m}}{\sqrt{n+z}} - 2 + \log \frac{n+z}{m}\geq 0$ where the equation holds if and only if $m=n+z$, which concludes the proof for (i). Now we give the proof for (ii). Since as $m\rightarrow\infty$, we have $W_2(\pi_m^f,\pi_n^c)\geq \left(\frac{\sigma}{\sqrt{m}}-\frac{\sigma}{\sqrt{n+z}}\right)^2 \rightarrow \frac{\sigma^2}{n+z}$, and that $\left|\frac{1}{m} \left[2\frac{\sqrt{m}}{\sqrt{n+z}} - 2 + \log \frac{n+z}{m}\right]\right|=\frac{1}{m}O(\sqrt{m})\rightarrow 0$. Thus the middle part expression in \eqref{eqn:KL_l2_gaussian} is asymptotically equivalent to $\frac{m}{2\sigma^2}$.
\end{proof}
\begin{remark}
Let $\theta_1 = (w_n\overline{y}_n + (1-w_n) \mu_0, \frac{\sigma}{\sqrt{n+z}})$ and $\theta_2 = (\overline{x}_m, \frac{\sigma}{\sqrt{m}})$ be the parameters for the two Gaussian posteriors given in~\eqref{eqn:gaussianposteriors}, then $W_2 (\pi_m^f,\pi_n^c) = ||\theta_1-\theta_2||_2^2$, where $||\cdot||_2^2$ is the $L^2$ norm.
The statements given in proposition~\ref{proposition:KLandFisherinfor_Gaussian} corresponds to the general result from information theory: for a parametric family $\{p_\theta(x)\}$, we have
\begin{equation*}
\lim_{\theta'\rightarrow\theta}\frac{1}{(\theta-\theta')^2} {\rm KL} (p_\theta, p_{\theta'}) = \frac{1}{\ln 4} J(\theta),
\end{equation*}
where $J(\theta)$ is the Fisher information. See Exercise 7 on page 334 of~\citet{cover2012elements} for details.
\end{remark}

\section{Details for Regression Model in Section~\lowercase{\ref{subsec:linearregression}}}
\label{appendix:expressions_linear_regression}

Here we give the expressions for the posterior distributions $\{\pi_u^f(\boldsymbol{\beta}), \pi_u^c(\boldsymbol{\beta})\}$ and the Wasserstein distances $\{W_2(\pi_m^f, \pi_n^c), W_2(\pi_m^c),\pi_n^f\}$. 

For ease of notation, we define the following quantities, for $u \geq n$:
\begin{align*}
   \overline{x}_u & = \frac{1}{u} \sum_{i=1}^u x_i^{(u)},  \quad
   \overline{y}_u = \frac{1}{u} \sum_{i=1}^u y_i^{(u)},  \\
   \overline{x^2}_u & = \frac{1}{u} \sum_{i=1}^u ((x_i^\prime)^{(u)})^2,  \quad
   \overline{xy}_u = \frac{1}{u} \sum_{i=1}^u (x_i^\prime)^{(u)} y_i^{(u)}, \\
 \hat{\beta}_u & =
 \begin{bmatrix}
 \overline{y}_u \\
 \overline{xy}_u / \overline{x^2}_u
 \end{bmatrix}.
\end{align*}
Define $\boldsymbol{\eta}_u$ and $\Sigma_u$ as
\begin{align*}
   \boldsymbol{\eta}_u & = \begin{bmatrix}
   u + \frac{1}{\tau^2_1} & u\overline{x}_u \\
   u\overline{x}_u & u\overline{x^2}_u + \frac{1}{\tau^2_2}
   \end{bmatrix} ^ {-1}
   \begin{bmatrix}
   u\overline{y}_u + \frac{\mu_0}{\tau_1^2} \\ u\overline{x^2}_u \tfrac{\overline{xy}_u}{\overline{x^2}_u} + \frac{\gamma_0}{\tau_2^2}
   \end{bmatrix}, \\
   \Sigma_u & = 
   \sigma^{2}\lp\begin{bmatrix}
   u + \frac{1}{\tau_1^2} & u \overline{x}_u \\
   u \overline{x}_u &  u\overline{x^2}_u + \frac{1}{\tau_2^2} 
   \end{bmatrix}\rp^{-1}.
\end{align*}
Then $\pi_u^c(\boldsymbol{\beta}) = \mathcal{N}\lp\boldsymbol{\eta}_u, \Sigma_u \rp$
while $\pi_u^f(\boldsymbol{\beta}) = \mathcal{N}
       \lp \hat{\boldsymbol{\beta}}_u, \hat{\Sigma}_u\rp$, where
\begin{align*}
 \hat{\boldsymbol{\beta}}_u & = \begin{bmatrix}
   \overline{y}_u \\ \tfrac{\overline{xy}_u}{\overline{x^2}_u}
   \end{bmatrix}, \quad 
   \hat{\Sigma}_u = \sigma^{2}\begin{bmatrix}
   u^{-1} & u\overline{x}_u \\
   u\overline{x}_u &  \frac{1}{u \overline{x^2}_u}
   \end{bmatrix}.
\end{align*}
Since for multivariate Gaussians $\nu_A$, $\nu_B$ with mean vectors $\boldsymbol{\mu}_A, \boldsymbol{\mu}_B \in \mathbb{R}^d$ and covariance matrices
$\Sigma_A, \Sigma_B \in \mathbb{R}^{d\times d}$, $W_2(\nu_A, \nu_B)$ is
\begin{align*}
   \lp\vectornorm{\boldsymbol{\mu}_A - \boldsymbol{\mu}_B}_2^2 + \text{tr}\lp \Sigma_A + \Sigma_B - 2(\Sigma_B^{\frac{1}{2}} \Sigma_A \Sigma_B^{\frac{1}{2}})^{\frac{1}{2}}\rp\rp^{\frac{1}{2}},
\end{align*}
we can write $W_2(\pi_m^f, \pi_n^c)$ and $W_2(\pi_m^c, \pi_n^f)$ in terms of posterior quantities $\hat{\boldsymbol{\beta}}_u, \boldsymbol{\eta}_u, \hat{\Sigma}_u$, and $\Sigma_u$ as follows:
\begin{align*}
    W_2(\pi_m^f, \pi_n^c) & = \lp\vectornorm{\hat{\boldsymbol{\beta}}_m - \boldsymbol{\eta}_n}_2^2 + \text{tr}\lp \hat{\Sigma}_m + \Sigma_n - 2(\Sigma_n^{\frac{1}{2}} \hat{\Sigma}_m \Sigma_n^{\frac{1}{2}})^{\frac{1}{2}}\rp\rp^{\frac{1}{2}},
\end{align*}
\begin{align*}
    W_2(\pi_m^c, \pi_n^f) & = \lp\vectornorm{\boldsymbol{\eta}_m - \hat{\boldsymbol{\beta}}_n}_2^2 + \text{tr}\lp \Sigma_m + \hat{\Sigma}_n - 2(\hat{\Sigma}_n^{\frac{1}{2}} \Sigma_m \hat{\Sigma}_n^{\frac{1}{2}})^{\frac{1}{2}}\rp\rp^{\frac{1}{2}}.
\end{align*}

\end{document}